%% file: Revisiting_Carr.tex
\documentclass[11pt]{article}
\usepackage[T1]{fontenc}
\usepackage[utf8]{inputenc}
\usepackage[english]{babel}
\usepackage{microtype}					
\usepackage{authblk}					
\usepackage{lipsum}						
\usepackage{color}

\usepackage{geometry}

\usepackage{booktabs}															
\usepackage{multirow}															
\usepackage{caption}	      													
\usepackage{siunitx}															
\usepackage{graphicx}		 													

\usepackage{amsmath}
\usepackage{amssymb}
\usepackage{mathtools}					
\usepackage{amsthm}						
\usepackage{dsfont}						  

\def\mathbi#1{\textbf{\itshape #1}}			

\theoremstyle{plain}
\newtheorem{theorem}{Theorem}[section]
\newtheorem{proposition}[theorem]{Proposition}
\newtheorem{lemma}[theorem]{Lemma}
\newtheorem{corollary}[theorem]{Corollary}

\theoremstyle{definition}
\newtheorem{definition}[theorem]{Definition}
\newtheorem{assumption}[theorem]{Assumption}
\newtheorem{example}[theorem]{Example}

\newtheorem{remark}[theorem]{Remark}

\numberwithin{equation}{section}

\DeclareMathOperator{\sgn}{sgn}

\usepackage[title,toc,titletoc,page]{appendix}									

\usepackage[autostyle,italian=guillemets]{csquotes} 
\usepackage[bibstyle=authoryear,citestyle=authoryear-comp, backend=biber]{biblatex}
\usepackage{hyperref}															
\newcommand{\mail}[1]{\href{mailto:#1}{\texttt{#1}}}							

\begin{document}

\title{Revisiting \\ the Implied Remaining Variance framework \\ of Carr and Sun (2014): \\ Locally consistent dynamics and sandwiched martingales}
\author[1]{Claude Martini}
\author[2]{Iacopo Raffaelli\thanks{Corresponding author. Email: \mail{iacopo.raffaelli@sns.it} \\ We thank Maria Elvira Mancino for her valuable help.}}
\affil[1]{\textit{Zeliade Systems, Paris, France}}
\affil[2]{\textit{Scuola Normale Superiore, Pisa, Italy}}
\maketitle

\input{Abstract/Abstract}
\input{Introduction/Introduction}
\input{IRV_definition/IRV_definition}
\input{Summary_Carr/Summary_Carr}
\input{Master_SDE/Master_SDE}
\input{Examples/Examples}
\input{Sandwiched_martingales/Sandwiched_martingales}
\input{Conclusion/Conclusion}

\newpage
\input{Appendix/Appendix}

\newpage
\printbibliography[heading=bibintoc]

\end{document}

%% file: Abstract/Abstract.tex
\begin{abstract}

\noindent Implied volatility is at the very core of modern finance, notwithstanding standard option pricing models continue to derive option prices starting from the joint dynamics of the underlying asset price and the spot volatility. These models often causes difficulties: no closed formulas for prices, demanding calibration techniques, unclear maps between spot and implied volatility. Inspired by the practice of using implied volatility as quoting system for option prices, models for the joint dynamics of the underlying asset price and the implied volatility have been proposed to replace standard option pricing models. Starting from~\textcite{CarrSun:2014}, we develop a framework based on the Implied Remaining Variance where minimal conditions for absence of arbitrage are identified, and smile bubbles are dealt with. The key concepts arising from the new IRV framework are those of \emph{locally consistent dynamics} and \emph{sandwiched martingale}. Within the new IRV framework, the results of~\textcite{SchweizerWissel:2008b} are reformulated, while those of~\textcite{ElAmraniJacquierMartini:2020} are independently derived.
\medskip


\end{abstract}

%% file: Introduction/Introduction.tex
\section{Introduction}
\label{sec:Introduction}

Today it is impossible to talk about financial markets without referring to implied volatility. Implied volatility computed from the Black-Scholes-Merton model is indeed a unifying and homogeneous quoting system for option prices. Although the assumptions made by~\textcite{BlackScholes:1973} and~\textcite{Merton:1973} have appeared to be too simplifying, investors see the BSM implied volatility as a functional transformation to improve quote stability and to highlight the information contained in option contracts.  Implied volatility has also proven to be an effective tool for estimating volatility risk premia and investors' risk aversion, like in the papers by~\textcite{ChernovGhysels:2000},~\textcite{Jackwerth:2000},~\textcite{AitSahaliaLo:2000},~\textcite{CarrWu:2009} and~\textcite{BollersleveGibsonZhou:2011}, among many others. 

Even though implied volatility is at the very core of modern ﬁnance, standard option pricing models continue to derive option prices starting from the joint dynamics of the underlying asset price and the spot volatility and then exploiting no-arbitrage arguments (\emph{martingale approach}). We recall the~\textcite{Heston:1993}  model for equity and the SABR model by~\textcite{Haganetal:2002} for interest rates, whereas~\textcite{AlosMancinoWang:2019} provided a comprehensive and updated review. Despite the extensive use of these models,  the related difficulties are many: closed formulas for option prices are not available, implementation requires ad hoc numerical techniques, and the map between spot and implied volatility is often not well-determined. From a market perspective, investors cannot observe spot volatilities, but they trade based on implied volatilities, which are available for a broad spectrum of strike prices and maturities.

For these reasons, the idea of directly modeling the joint dynamics of the underlying asset price and the implied volatility appears so natural. We note that this is equivalent to specify the joint dynamics of the prices of all the assets traded on the market, i.e., the underlying asset and the options. These \emph{market models} can easily fit well-known stylized facts from ﬁnancial time series and make perfect calibration possible. The same intuition led~\textcite{HJM:1992} to develop their framework for interest rates. 

Guaranteeing absence of arbitrage is the true challenge for the new models. While the market is arbitrage-free by construction adopting the martingale approach, market models require to identify no-arbitrage conditions. In the case of implied volatility models, we have to distinguish between static and dynamic arbitrage. If the prices of all the traded assets are martingales under an equivalent martingale measure, then the market is free of dynamic arbitrage. It corresponds to identify no-drift conditions on the joint dynamics, as it is in the Heath-Jarrow-Morton framework.  However, absence of static arbitrage imposes additional restrictions on how prices can jointly evolve, further constraining the processes' state space.

Due to \emph{``the extremely awkward set of conditions required for absence of arbitrage’'} (\cite{DavisObloj:2008}), the first attempts of implied volatility models (e.g.,~\cite{Lyons:1997},~\cite{AvellanedaZhu:1998},~\cite{LedoitSantaClara:1998},~\cite{Schonbucher:1999},~\cite{Hafner:2004},~\cite{Fengler:2005},~\cite{DaglishHullSuo:2007},~\cite{SunNiu:2016}) have been \emph{``largely unsuccessful''} (\cite{CarrWu:2016}). General conditions for the existence of arbitrage-free implied volatility models were derived for the first time by~\citeauthor{SchweizerWissel:2008a} (\citeyear{SchweizerWissel:2008a},~\citeyear{SchweizerWissel:2008b}).~\textcite{Babbar:2001} developed examples for a ﬁxed strike price and Bessel processes. To the best of our knowledge, these are the only positive results for continuous processes.

The paper~\emph{\citetitle{CarrSun:2014}} by~\textcite{CarrSun:2014} pertains to this literature, and we discuss it in detail in Section~\ref{sec:Carr}. Brieﬂy, the key-point is that~\citeauthor{CarrSun:2014} suggested replacing the implied volatility with the Implied Remaining Variance (IRV, hereafter), i.e., the square of the implied volatility multiplied by the time to maturity. Modeling the joint dynamics of the underlying asset price and the IRV, they were able to identify a no-drift condition on the IRV dynamics without making any assumption on the spot volatility dynamics.

\textcite{ElAmraniJacquierMartini:2020} resorted to IRV, too. In particular, they introduced the concept of \emph{smile bubble}: a joint dynamics such that, locally in time (i.e., until some time $\tau$ before the common maturity of the options), the market is arbitrage-free, although option prices are no longer equal to the expectation of the discounted ﬁnal payoﬀs under an equivalent martingale measure. At $\tau$, a change of regime occurs.

Our paper aims to formalize and extend the findings of~\textcite{CarrSun:2014}. The framework we develop allows for much more general IRV dynamics, keeping the original paper's advantages. We identify a necessary and sufficient condition for the option prices to be martingales in the form of a no-drift condition on the IRV dynamics. The set of solutions of the related stochastic differential equation is nothing else than the set of the IRV processes consistent with absence of dynamic arbitrage. We call this SDE the IRV master SDE. 


While~\textcite{CarrSun:2014} did not deal with static arbitrage, the new IRV framework faces both the arbitrage types. The distinction between locally and globally consistent dynamics is the way: if the IRV dynamics is locally consistent, then it ensures absence of dynamic but not static arbitrage; instead, the market is free of (both types of) arbitrage if the IRV process follows a globally consistent dynamics. We say that an IRV process follows a locally consistent dynamics if it satisfies the IRV master SDE; a solution of the IRV master SDE is said to follow a globally consistent dynamics if it vanishes at maturity. In Section~\ref{sec:examples}, we show how the results of~\textcite{SchweizerWissel:2008b} ensure the existence of solutions of the IRV master SDE vanishing at maturity, also providing an explicit family of solutions.

Through locally consistent dynamics, the new IRV framework retrieves also smile bubbles and the other findings of~\textcite{ElAmraniJacquierMartini:2020}.

Finally, we restate the results on IRV in terms of option prices. Indeed, we introduce the concept of sandwiched martingale, and we show the equivalence between an option price described by a sandwiched martingale and an IRV process following a locally consistent dynamics. Examples of sandwiched martingales are provided.

The paper is organized as follows: in Section~\ref{sec:IRVdef}, we give the definition of Implied Remaining Variance, and we set the market environment; in Section~\ref{sec:Carr}, we explain the main advantages and issues of~\textcite{CarrSun:2014}; Section~\ref{sec:masterSDE} contains the results about locally (and globally) consistent dynamics and absence of arbitrage; in Section~\ref{sec:examples}, we provide examples of solutions of the IRV master SDE; in Section~\ref{sec:sandwiched}, we deal with sandwiched martingales; Section~\ref{sec:conclusion} concludes.

%% file: IRV_definition/IRV_definition.tex
\section{Definition of Implied Remaining Variance}
\label{sec:IRVdef}

In this section, we recall the definition of Implied Remaining Variance, and we set the notation for the rest of the paper.

Due to the frequent use of the Black-Scholes option pricing formulas, the following two functions considerably simplify the notation: 
\begin{definition}
	The function $BS: \mathbb{R} \times [0,+\infty] \rightarrow [0,1]$ is defined as
	\[
		BS(k,v) \coloneqq 
		\begin{cases}
			1 & \text{if $v = + \infty$} \\
			\mathcal{N}\big(d_+(k,v)\big) - e^k \cdot \mathcal{N}\big(d_-(k,v)\big) & \text{if $v \in (0,+\infty)$} \\
			\big(1 - e^k \big)_+ & \text{if $v = 0$}
		\end{cases}
	\]
	and the function ${BS}_P: \mathbb{R} \times [0,+\infty] \rightarrow [0,+\infty)$ is defined as 
	\[
		{BS}_P(k,v) \coloneqq 
		\begin{cases}
			e^k & \text{if $v = +\infty$} \\
			e^k \cdot \mathcal{N}\big(- d_-(k,v)\big) - \mathcal{N}\big(- d_+(k,v)\big)  & \text{if $v \in (0,+\infty)$} \\
			\big(e^k - 1\big)_+ & \text{if $v = 0$}
		\end{cases}
	\]
	where $\mathcal{N}(\cdot)$ is the Gaussian cumulative distribution function and $d_\pm(k,v) \coloneqq \frac{-k}{v} \pm \frac{v}{2}$.
\end{definition}

Let $\big(\Omega,\mathcal{F},\mathbb{Q}\big)$ be a probability space describing a perfect market with zero interest rates under the risk-neutral measure. A financial asset is traded whose price is almost surely given by the usual exponential process
\[
\begin{cases}
	S_t = S_0 \cdot \exp\Big\{\int_0^t \sigma_s \, dW_s - \frac{1}{2} \int_0^t \sigma_s^2 \, ds \Big\} \\
	S_0  > 0
\end{cases}
\]
for any $t \in [0,+\infty)$, where $\sigma$ is an a.s. positive process and $W$ is a standard Brownian motion. 

For convenience, we recall that the Black-Scholes model requires $\sigma_t \equiv \sigma > 0$, so that the Black-Scholes prices of a European call and put option with this asset as underlying, strike price $K$ and maturity $T$, are given by
\begin{align*}
	C^{BS}_t(K,T) & = S_t \cdot BS\Big( k_t, \sigma \sqrt{T-t} \Big) \\
	P^{BS}_t(K,T) & = S_t \cdot {BS}_P\Big( k_t, \sigma \sqrt{T-t} \Big)
\end{align*}
for any $t \in [0,T)$, where $k_t\coloneqq \log \big(K / S_t \big)$ is the log-moneyness process.

In general, assume to observe on the market the prices of both the asset and the options at any time $t \in [0,T)$. The \emph{Implied Remaining Variance (IRV)} corresponding to the European call option with strike price $K$ and maturity $T$, denoted as $\omega_t(K,T)$, is then implicitly defined by the sequence of equalities:
\[
C^{obs}_t(K,T) = C^{BS}_t(K,T) = S_t \cdot BS\Big(k_t, \sqrt{\omega_t(K,T)}\Big).
\]
In other words, denoting the Black-Scholes implied volatility with $I_t(K,T)$, it is:
\[
\omega_t(K,T) = I_t^2(K,T) \cdot \big(T-t\big).
\]

A fundamental assumption we make is:
\begin{assumption} \label{ass:PCP}
	European call and put options with identical strike price and maturity share the same IRV, that is:
	\[
	C^{obs}_t(K,T) = S_t \cdot BS\Big(k_t, \sqrt{\omega_t(K,T)}\Big) \quad \Longleftrightarrow \quad P^{obs}_t(K,T) = S_t \cdot {BS}_P\Big(k_t, \sqrt{\omega_t(K,T)} \Big)
	\]
	for any $t \in [0,T)$. Equivalently, the \emph{Put-Call-Parity} holds on the market:
	\[
	P^{obs}_t(K,T) + S_t = C^{obs}_t(K,T) + K
	\]
	for any $t \in [0,T)$.
\end{assumption}

%% file: Summary_Carr/Summary_Carr.tex
\section{The IRV framework by Carr and Sun (2014)} 
\label{sec:Carr}

This section deals with the IRV framework suggested by~\citeauthor{CarrSun:2014} in~\emph{\citetitle{CarrSun:2014}} (\citeyear{CarrSun:2014}, Journal of Fixed Income). We show how this framework can represent an alternative to standard option pricing models (subsection~\ref{subsec:CarrOK}) and the related issues (subsection~\ref{subsec:CarrNO}).


\subsection{Summary of Carr and Sun (2014)}
\label{subsec:CarrOK}

Given two time stamps $0 < T < T'$,~\citeauthor{CarrSun:2014} considered a call swaption maturing at $T$, written on an underlying swap maturing at $T'$. Let $F_t(T,T')$ be the underlying forward swap rate at time $t \in [0,T]$. At time $T$, the call swaption's payoff in units of the forward-starting annuity is $(F_T - K)_+$, where $K > 0$ is the strike price and the arguments $(T,T')$ of $F$ have been omitted because fixed. For the same reason, we write $\omega(K)$ in the following.

Under the forward swap measure $\mathbb{Q}$ (see~\cite{Wu:2009}),~\citeauthor{CarrSun:2014} proposed to replace the stochastic volatility model
\[
\begin{cases}
	dF_t = \sigma_t F_t \, dW_t \\
	d\sigma_t = \alpha(\sigma_t) \, dt + \beta(\sigma_t) \, dH_t \\
	d \langle W, H \rangle_t = r \, dt
\end{cases}
\]
where $\alpha(\sigma)$ and $\beta(\sigma)$ are real-valued functions, $r \in [-1,1]$ is a constant, and $W$ and $H$ are standard Brownian motions, with the \emph{IRV model}
\begin{equation}\label{eq:CSIRV}
\begin{cases}
	dF_t = \sigma_t F_t \, dW_t \\
	d\omega_t(K) = a\big(\omega_t(K)\big) \sigma_t^2 \, dt + b\big(\omega_t(K)\big) \sigma_t \, dZ_t \\
	d\langle W, Z \rangle_t = \rho \, dt
\end{cases}
\end{equation}
where $a(\omega)$, $b(\omega)$ are two real-valued deterministic functions of the current IRV level, $\rho \in [-1,1]$ is a constant, and $W$ and $Z$ are standard Brownian motions.

In an arbitrage-free market, the call swaption price process, $C(K)$, must be a $\mathbb{Q}$-martingale. By definition of IRV, it is
\[
C_t(K) = F_t \cdot \mathcal{N}\Big(d_+\big(k_t,\sqrt{\omega_t(K)}\big)\Big) - K \cdot \mathcal{N}\Big(d_-\big(k_t,\sqrt{\omega_t(K)}\big)\Big)
\]
for any $t \in [0,T)$, where $k_t \coloneqq \log(K/F_t)$. Applying It\^o formula to $C(K)$,~\citeauthor{CarrSun:2014} derived the no-drift condition:
\begin{equation}\label{eq:CSnodrift}
	a(\omega_t(K)) = - \Bigg\{ 1 + \frac{\rho  \, b(\omega_t(K))}{2} - \bigg(\frac{1}{\omega_t(K)} + \frac{1}{4} \bigg) \cdot \frac{b^2(\omega_t(K))}{4} + \frac{\rho \, k_t \,  b(\omega_t(K))}{\omega_t(K)} + \frac{k_t^2  \, b^2(\omega_t(K))}{4 \, \omega_t^2(K)} \Bigg\}.
\end{equation}
Hence, the call swaption price is a $\mathbb{Q}$-martingale only if the drift coefficient in the IRV dynamics satisfies Eq.~\eqref{eq:CSnodrift}. 

A remarkable feature of this no-drift condition is that it does not depend on the model chosen for the spot volatility. Indeed, the spot volatility dynamics has been left unspecified. Moreover,~Eq.~\eqref{eq:CSnodrift} provides modelers with the drift coefficient ensuring call swaption price martingality, once speciﬁed the IRV diﬀusion. The analogy with the~\textcite{HJM:1992} framework for interest rates is clear.

On the contrary,~\citeauthor{CarrSun:2014} decided to specify both the drift and diﬀusion terms and thus determine the IRV smile. 
In particular, they chose the following functional forms for $a(\omega)$ and $b(\omega)$:
\begin{align}
	\label{eq:CSa} a(\omega) & = -a_1 \omega + a_0 -1 \\
	\label{eq:CSb} b(\omega) & = \omega
\end{align}
where $a_0 \in \mathbb{R}$ and $a_1 > 0$ are constant. The related IRV model is:
\begin{equation}\label{eq:CSdw}
\begin{cases}
	dF_t = \sigma_t F_t \, dW_t \\
	d\omega_t(K) = - \Big[1 - a_0 + a_1 \, \omega_t(K)\Big] \sigma_t^2 \, dt + \omega_t(K) \, \sigma_t \, dZ_t \\
	d\langle W, Z\rangle_t = \rho \, dt.
\end{cases}
\end{equation}
Substituting Eq.~\eqref{eq:CSa}--\eqref{eq:CSb} into Eq.~\eqref{eq:CSnodrift} leads to the quadratic equation
\[
\omega_t^2(K) + 4 \Big(1 - 2 \rho + 4 a_1 \Big) \cdot \omega_t(K) - 4 \Big( 4 a_0 + 4 \rho k_t + k_t^2 \Big) = 0.
\]
The discriminant is
\[
\Delta_t(K) = 16 \Big[k_t^2 + 4 \rho k_t + 4 a_0 + \big(1 - 2 \rho + 4 a_1\big)^2 \Big]
\]
which depends on $K$ through $k_t$ and is almost surely positive under the assumption 
\[
a_0 > \rho^2.
\]
Assuming also
\[
1 - 2 \rho + 4 a_1 \ge 0,
\]
the unique positive solution is\footnote{Eq.~\eqref{eq:CSrightsol} corrects Eq. (40) in~\textcite{CarrSun:2014}.}
\begin{equation}\label{eq:CSrightsol}
	\omega_t(K) = 2 \Big(2 \rho - 1 - 4 a_1 + \sqrt{D_t(K)} \Big)
\end{equation}
where $D_t(K) = \Delta_t(K) / 16 > 0$.


\subsection{Issues with Carr and Sun (2014)}
\label{subsec:CarrNO}

At this point, we investigate whether the expression for the IRV smile in Eq.~\eqref{eq:CSrightsol} is consistent with the dynamics in Eq.~\eqref{eq:CSdw}. The answer we find is negative. 

Indeed, applying It\^o formula to Eq.\eqref{eq:CSrightsol}, we obtain
\begin{equation}\label{eq:CSdw2}
	d\omega_t = \frac{1}{\sqrt{D_t}} \, \bigg[ 1 + 2 \rho + k_t - \frac{(2 \rho + k_t)^2}{D_t} \bigg] \, \sigma_t^2 \, dt - 2 \, \frac{2 \rho + k_t}{\sqrt{D_t}} \, \sigma_t \, dW_t
\end{equation}
where the argument $(K)$ has been suppressed for notational simplicity. 

To compare Eq.~\eqref{eq:CSdw} and~\eqref{eq:CSdw2}, we decompose
\[
dZ_t = \rho \, dW_t + \sqrt{1- \rho^2} \, dW_t^{\bot}
\]
or equivalently, under the assumption $\rho \neq 0$ (if $\rho = 0$, Eq.~\eqref{eq:CSdw} and~\eqref{eq:CSdw2} cannot be consistent because driven by two independent Brownian motions),
\[
dW_t = \frac{1}{\rho} \, dZ_t - \sqrt{\frac{1}{\rho^2}-1} \, dW_t^{\bot}
\]
and we rewrite Eq.~\eqref{eq:CSdw2} as 
\[
	d\omega_t = \frac{1}{\sqrt{D_t}} \, \bigg[ 1 + 2 \rho + k_t - \frac{(2 \rho + k_t)^2}{D_t} \bigg] \, \sigma_t^2 \, dt - 2 \, \frac{2 \rho + k_t}{\sqrt{D_t}} \, \sigma_t \, \bigg( \frac{1}{\rho} \, dZ_t - \sqrt{\frac{1}{\rho^2}-1} \, dW_t^{\bot} \bigg).
\]

The last dynamics can be consistent with that in Eq.~\eqref{eq:CSdw} only if 
\[
2 \, \frac{2 \rho + k_t}{\sqrt{D_t}} \, \sigma_t \cdot \sqrt{\frac{1}{\rho^2}-1} = 0
\]
that is only if:
\begin{itemize}
	\item either $\rho = - k_t / 2$, which contradicts the assumptions that $\rho$ is a constant (unless $F$ is constant over time) and $k_t \in \mathbb{R}$ (not necessarily in $[-2,2]$);
	\item or $\rho = \pm 1$.
\end{itemize}

We start by considering the case $\rho = -1$. Under the assumptions $a_0 > 1$ and $a_1 \ge - 3/4$, we should have 
\[
\begin{cases}
	d\omega_t = - \big[1 - a_0 + a_1 \omega_t\big] \sigma_t^2 \, dt + \omega_t \sigma_t \, dZ_t \\
	d\omega_t = \frac{1}{\sqrt{D_t}} \Big[ -1 + k_t - \frac{(k_t - 2)^2}{D_t} \Big] \sigma_t^2 \, dt + 2 \frac{k_t - 2}{\sqrt{D_t}} \sigma_t \, dZ_t
\end{cases}
\]
which would imply
\[
\omega_t = 2 \, \frac{k_t - 2}{\sqrt{D_t}}
\]
and therefore
\[
\begin{cases}
	d\omega_t = - \big[1 - a_0 + a_1 \omega_t\big] \sigma_t^2 \, dt + \omega_t \sigma_t \, dZ_t \\
	d\omega_t = - \Big[ 1 - \Big(1 + \frac{1}{\sqrt{D_t}}\Big) - \frac{1}{2} \omega_t + \frac{1}{4 \sqrt{D_t}} \omega_t^2 \Big] \sigma_t^2 \, dt + \omega_t \sigma_t \, dZ_t.
\end{cases}
\]
The two dynamics would be consistent only if 
\[
\begin{cases}
	a_0 = 1 + \frac{1}{\sqrt{D_t}} \\
	a_1 = - \frac{1}{2} \omega_t \\
	\frac{1}{4 \sqrt{D_t}} = 0
\end{cases}
\]
but the last equality cannot hold.

In the case $\rho = +1$, the result is the same. Under the assumptions $a_0 > 1$ and $a_1 \ge 1/4$, we should have
\[
\begin{cases}
	d\omega_t = - \big[1 - a_0 + a_1 \omega_t\big] \sigma_t^2 \, dt + \omega_t \sigma_t \, dZ_t \\
	d\omega_t = \frac{1}{\sqrt{D_t}} \Big[ 3 + k_t - \frac{(k_t + 2)^2}{D_t} \Big] \sigma_t^2 \, dt - 2 \frac{k_t + 2}{\sqrt{D_t}} \sigma_t \, dZ_t
\end{cases}
\]
that is
\[
\omega_t = - 2 \, \frac{k_t + 2}{\sqrt{D_t}}
\]
and
\[
\begin{cases}
	d\omega_t = - \big[1 - a_0 + a_1 \omega_t\big] \sigma_t^2 \, dt + \omega_t \sigma_t \, dZ_t \\
	d\omega_t = - \Big[ 1 - \Big(1 + \frac{1}{\sqrt{D_t}}\Big) + \frac{1}{2} \omega_t + \frac{1}{4 \sqrt{D_t}} \omega_t^2 \Big] \sigma_t^2 \, dt + \omega_t \sigma_t \, dZ_t.
\end{cases}
\]
Again, the consistency of the two dynamics would require
\[
\begin{cases}
	a_0 = 1 + \frac{1}{\sqrt{D_t}} \\
	a_1 = \frac{1}{2} \\
	\frac{1}{4 \sqrt{D_t}} = 0
\end{cases}
\] 
but the last equality cannot hold. Therefore, we can conclude that the expression for the IRV smile in Eq.~\eqref{eq:CSrightsol} is not consistent with the dynamics set in Eq.~\eqref{eq:CSdw}. 

A more severe issue in the article by~\citeauthor{CarrSun:2014} is the following. The no-drift condition in Eq.~\eqref{eq:CSnodrift} is a necessary condition for the call swaption price to be a $\mathbb{Q}$-driftless process, that is, under certain integrability assumptions, a $\mathbb{Q}$-martingale. However, the swaption price martingality is not sufficient to ensure absence of arbitrage. In particular, in a market where swaptions with different strike prices are traded, static arbitrage opportunities continue to exist even though the swaptions' prices are martingales. This is a crucial difference with respect to the Heath-Jarrow-Morton framework for interest rates, and we deal with it in detail in the next section.

Finally, we recall that $a(\omega)$ and $b(\omega)$ in Eq.~\eqref{eq:CSIRV}~--~respectively, the drift and diffusion terms in the IRV dynamics~--~are supposed to be deterministic functions of the current IRV value: this is a somewhat restrictive assumption.

%% file: Master_SDE/Master_SDE.tex
\section{Locally consistent dynamics}
\label{sec:masterSDE}

This section investigates what minimal properties a stochastic process must satisfy to be an appropriate candidate as IRV process of an arbitrage-free market. At the same time, smile bubbles (\cite{ElAmraniJacquierMartini:2020}) are introduced into the new IRV framework.


\subsection{Single option market}

We start by considering a market on which one European call option, with strike price $K > 0$ and maturity $T > 0$, is traded. Being fixed, we omit the arguments $(K,T)$ of $C$ and $\omega$. We recall that $k_t \coloneqq \log(K/S_t)$, and we set the notation $X^{\tau}_t \coloneqq X_{t \wedge \tau}$, for a process $X$ and a stopping time $\tau$.

First, we prove the following lemma:
\begin{lemma}\label{lemma:locmart}
	On a filtered probability space $\big(\Omega,\mathcal{F},(\mathcal{F}_t)_{t \ge 0},\mathbb{Q}\big)$, let $S$ be a positive martingale and $\omega$ a process with values in $[0,+\infty]$, almost surely. Define the process $C_t^{BS} \coloneqq S_t \cdot BS\big(k_t,\sqrt{\omega_t}\big)$, for any $t \in [0,+\infty)$. Then $C^{BS}$ is a martingale if and only if it is a local martingale.
\end{lemma}
\begin{proof}
	Let $C^{BS}$ be a local martingale. Define the process $P_t^{BS} \coloneqq S_t \cdot {BS}_P\big(k_t,\sqrt{\omega_t}\big)$, for any $t \in [0,+\infty)$. By the Put-Call-Parity, it is $P^{BS}_t = C^{BS}_t + K - S_t$, for any $t \in [0,+\infty)$. Therefore, $P^{BS}$ is given by the sum of a local martingale, a constant and a martingale, and so is itself a local martingale. Since $P^{BS}$ is uniformly bounded by $K$, it is also a martingale. Applying the Put-Call-Parity again, $C^{BS}$ turns out to be equal to the sum of two martingales and a constant, hence is itself a martingale. Trivially, if $C^{BS}$ is a martingale, then it is a local martingale.
\end{proof}

Let $W$ and $Z$ be two independent Brownian motions on $\big(\Omega,\mathcal{F},(\mathcal{F}_t)_{t \ge 0},\mathbb{Q}\big)$. Moreover, assume the same space can accommodate two almost surely positive and finite random variables $S_0$ and $\omega_0$, independent of $\mathcal{F}^W_{\infty}$ and $\mathcal{F}^Z_{\infty}$. 
Finally, consider the left-continuous filtration
\[
\mathcal{G}_t \coloneqq \sigma\big(S_0,\omega_0,W_s,Z_s: \, s \in [0,t] \big), \qquad t \in [0,+\infty),
\]
as well as the collection of null sets
\[
\mathcal{H} \coloneqq \big\{ H \subseteq \Omega: \, \exists G \in \mathcal{G}_{\infty} \text{ with $H \subseteq G$ and $\mathbb{Q}(G) = 0$} \big\},
\]
and define the augmented filtration $(\mathcal{F}_t)_{t \ge 0}$ as
\[
\mathcal{F}_t \coloneqq \sigma\big( \mathcal{G}_t \cup \mathcal{H}\big), \quad t \in [0,+\infty); \qquad \mathcal{F}_{\infty} \coloneqq \sigma \bigg(\bigcup_{t \ge 0} \mathcal{F}_t \bigg).
\]

\begin{theorem}\label{thr:masterSDE}
	Consider the filtered probability space $\big(\Omega, \mathcal{F}, (\mathcal{F}_t)_{t \ge 0}, \mathbb{Q}\big)$ described above, where $(\mathcal{F}_t)_{t \ge 0}$ is the augmented filtration, $W$ and $Z$ the two independent Brownian motions, and $S_0$ and $\omega_0$ the two almost surely positive and finite random variables. Moreover:
	\begin{enumerate}
		\item $\sigma \in L_{loc}^2\big((0,+\infty)\big)$ is such that the process $S$, defined by
		\[
		S_t \coloneqq S_0 \cdot \exp \bigg\{ \int_0^t \sigma_s \, dW_s - \frac{1}{2} \int_0^t \sigma_s^2 \, ds \bigg\}, \qquad \text{a.s.}
		\] 
		for any $t \in [0,+\infty)$, is a martingale.
		\item $\omega$ is a progressively measurable process with values in $[0,+\infty]$ and continuous paths, a.s.
		\item For $n \ge 1$, define the following sequences of stopping times\footnote{We adopt the convention: $\inf\{\emptyset\} = +\infty$.}
		\begin{equation}\label{eq:stopping}
			\tau^{0}_n \coloneqq \inf \big\{t \ge 0 : \omega_t \le 1/n \big\}, \quad \tau^{\infty}_n \coloneqq \inf \big\{t \ge 0 : \omega_t \geq n \big\}, \quad \tau_n \coloneqq \tau^{\infty}_n \wedge \tau^0_n.
		\end{equation}
		\item For any $n \ge 1$, there exist progressively measurable processes $a$, $b$, $c$ such that
		\[
		\int_0^{t \wedge \tau_n} \Big\{ \big\lvert a_s\sigma_s^2 \big\rvert + \big\lvert b_s \omega_s \sigma_s \big\rvert^2 + \big\lvert c_s \omega_s \sigma_s \big\rvert^2 \Big\} \, ds < \infty, \qquad \text{a.s.}
		\]
		and
		\begin{equation}\label{eq:SDE}
			\omega_t^{\tau_n} = \omega_0 + \int_0^{t \wedge \tau_n} a_s \sigma_s^2 \, ds + \int_0^{t\wedge \tau_n} b_s \omega_s \sigma_s \, dZ_s + \int_0^{t \wedge \tau_n} c_s \omega_s \sigma_s \, dW_s, \qquad \text{a.s.}
		\end{equation}
		for any $t \in [0,+\infty)$.
		\item Define the process
		\begin{equation}\label{eq:C}
			C^{BS}_t \coloneqq S_t \cdot BS \big( k_t, \sqrt{\omega_t} \big),
		\end{equation}
		where $k_t \coloneqq \log(K / S_t)$, for any $t \in [0,+\infty)$.
	\end{enumerate}
	
	In this framework, $C^{BS, \tau_n}$ is a martingale if and only if the following no-drift condition holds:
	\begin{multline}\label{eq:nodrift}
		a_t^{\tau_n} = \Bigg(\frac{{(b_t^{\tau_n}})^2 + ({c_t^{\tau_n}})^2}{16}\Bigg) \, ({\omega_t^{\tau_n}})^2 + \Bigg(\frac{({b_t^{\tau_n}})^2 + ({c_t^{\tau_n}})^2- 2 {c_t^{\tau_n}}}{4}\Bigg) \, \omega_t^{\tau_n} \\
		- \Bigg[\frac{({b_t^{\tau_n}})^2 + ({c_t^{\tau_n}})^2}{4} ({k_t^{\tau_n}})^2 + c_t^{\tau_n} \, k_t^{\tau_n} + 1\Bigg], \quad \text{a.s.}
	\end{multline}
	for any $t \in [0,+\infty)$, $n \ge 1$.
\end{theorem}

\begin{proof}
	For any $n \ge 1$ and $t \in [0,+\infty)$, it is
	\begin{equation*}
		C^{BS, \tau_n}_t = S_t^{\tau_n} \cdot \mathcal{N}\Big(d_+\Big(k_t^{\tau_n},\sqrt{\omega_t^{\tau_n}}\Big)\Big) - K \cdot \mathcal{N}\Big( d_- \Big(k_t^{\tau_n},\sqrt{\omega_t^{\tau_n}}\Big)\Big), \qquad a.s.
	\end{equation*}
	Suppress the arguments of $d_{\pm}$ for notational simplicity and apply the It\^o formula to $C^{BS,\tau_n}$:
	\[
	\begin{split}
		C^{BS,\tau_n}_t = C^{BS}_0 & + \int_0^{t \wedge \tau_n} b_s \frac{K \sqrt{\omega_s} \, \varphi(d_-)}{2} \sigma_s \, dZ_s + \int_0^{t \wedge \tau_n}\Bigg[ S_s \, \mathcal{N}(d_+) + c_s \frac{ K \sqrt{\omega_s} \, \varphi(d_-)}{2} \Bigg] \sigma_s \, dW_s \\	
		& + \int_0^{t \wedge \tau_n}\Bigg[ a_s \frac{K \, \varphi(d_-)}{2 \sqrt{\omega_s}} + c_s \frac{K \big(k_s + \frac{\omega_s}{2}\big) \, \varphi(d_-) }{2\sqrt{\omega_s}} \\
		& + \big( b_s^2 + c_s^2 \big)  \frac{K \Big(k^2_s- \frac{\omega_s^2}{4} - \omega_s \Big) \, \varphi(d_-) }{8 \sqrt{\omega_s}} + \frac{S_s \, \varphi(d_+)}{2 \sqrt{\omega_s}} \Bigg] \sigma_s^2 \, ds, \qquad a.s.
	\end{split}
	\]
	
	Assuming $C^{BS,\tau_n}$ is a martingale, the drift term in the equation above must be $0$. That is
	\begin{multline*}
		\Bigg[ a_t^{\tau_n} \frac{K \, \varphi(d_-)}{2 \sqrt{\omega_t^{\tau_n}}} + c_t^{\tau_n} \frac{K  \Big(k_t^{\tau_n} + \frac{\omega_t^{\tau_n}}{2}\Big) \, \varphi(d_-) }{2\sqrt{\omega_t^{\tau_n}}} \\
		+ \Big( ({b_t^{\tau_n}})^2 + ({c_t^{\tau_n}})^2 \Big)  \frac{K \Big(({k_t^{\tau_n}})^2 - \frac{({\omega_t^{\tau_n}})^2}{4} - \omega_t^{\tau_n} \Big) \, \varphi(d_-) }{8 \sqrt{\omega_t^{\tau_n}}} + \frac{S_t^{\tau_n} \, \varphi(d_+)}{2 \sqrt{\omega_t^{\tau_n}}} \Bigg] ({\sigma_t^{\tau_n}})^2 = 0, \quad a.s.
	\end{multline*}
	which is equivalent to
	\begin{multline*}
		a_t^{\tau_n} = \Bigg(\frac{({b_t^{\tau_n}})^2 + ({c_t^{\tau_n}})^2}{16}\Bigg) \, ({\omega_t^{\tau_n}})^2 + \Bigg(\frac{({b_t^{\tau_n}})^2 + ({c_t^{\tau_n}})^2 -2 c_t^{\tau_n}}{4} \Bigg) \, \omega_t^{\tau_n} \\
		- \Bigg[\frac{({b_t^{\tau_n}})^2 +( {c_t^{\tau_n}})^2}{4} \, ({k_t^{\tau_n}})^2 + c_t^{\tau_n} \, k_t^{\tau_n} + 1 \Bigg], \qquad a.s.
	\end{multline*}

	Assuming instead that $a^{\tau_n}$ satisfies Eq.~\eqref{eq:nodrift}, we have
	\begin{multline*}
		C^{BS,\tau_n}_t = C_0^{BS} + \int_0^{t \wedge \tau_n} \frac{K b_s \sqrt{\omega_s} \sigma_s \, \varphi(d_-)}{2} \, dZ_s \\
		+ \int_0^{t \wedge \tau_n}\Bigg[ S_s \, \mathcal{N}(d_+) + \frac{ K  c_s \sqrt{\omega_s} \, \varphi(d_-)}{2} \Bigg] \sigma_s \, dW_s, \qquad a.s.
	\end{multline*}
	Three observations follow. First,
	\[
	\int_0^{t \wedge \tau_n} S_s^2  \sigma_s^2  \mathcal{N}^2(d_+) \, ds < \int_0^{t \wedge \tau_n} S_s^2  \sigma_s^2 \, ds = \langle S \rangle_{t \wedge \tau_n}, \qquad a.s.
	\]
	Since $S$ is a continuous martingale, it is almost surely $\int_0^{t \wedge \tau_n} S_s^2 \sigma_s^2 \, ds < \infty$. Therefore, 
	\[
	{S^{\tau_n}}\, {\sigma^{\tau_n}} \, \mathcal{N}(d_+) \in L^2_{loc}\big((0,+\infty)).
	\] 
	Second, 
	\[
	\begin{split}
		\int_0^{t \wedge \tau_n} \frac{ K^2  c_s^2 \omega_s \sigma^2_s \, \varphi^2(d_-)}{4}  \, ds & \le \frac{K^2}{8 \pi}  \int_0^{t \wedge \tau_n} \frac{c_s^2 \omega^2_s \sigma^2_s}{\omega_s} \, ds \le \frac{n K^2}{8 \pi}  \int_0^{t \wedge \tau_n} c_s^2 \omega^2_s \sigma^2_s \, ds < \infty, \quad a.s.
	\end{split}
	\]
	Hence, 
	\[
	\sqrt{\omega^{\tau_n} } \, {c^{\tau_n}} \, {\sigma^{\tau_n}}  \, \varphi(d_-) \in L^2_{loc}\big([0,+\infty)\big).
	\]
	Third, in the same manner it can be proved that 
	\[
	\sqrt{\omega^{\tau_n}} \, {b^{\tau_n}} \, {\sigma^{\tau_n}} \,  \varphi(d_-) \in L^2_{loc}\big([0,+\infty)\big).
	\]
	Therefore, $C^{BS, \tau_n}$ is a local martingale because it is equal to the sum of a constant and three local martingales. By Lemma~\ref{lemma:locmart}, $C^{BS, \tau_n}$ is also a martingale.
\end{proof}

Within the framework outlined in Theorem~\ref{thr:masterSDE}, we give the following definitions.

\begin{definition}
	We say that $\omega$ follows a \emph{locally consistent dynamics} if, for any $n \ge 1$, there exist $a$, $b$ and $c$ such that it satisfies the dynamics in Eq.~\eqref{eq:SDE} and the no-drift condition in Eq.~\eqref{eq:nodrift} simultaneously.
\end{definition}

\begin{definition}
	The process $\omega$ is said to be the market \emph{IRV process} if, for any $n \ge 1$, the call price process stopped at $\tau_n$, denoted as $C^{\tau_n}$, and the process $C^{BS,\tau_n}$ defined in Eq.~\eqref{eq:C} are such that $\mathbb{Q}\big\{ C^{\tau_n}_t = C^{BS,\tau_n}_t; \, t \in [0,T) \big\}$.
\end{definition}

Assuming $\omega$ is the market IRV process, the stopped call price is a martingale if and only if $\omega$ follows a locally consistent dynamics. In particular, it means that, for any $n \ge 1$, there must be $a$, $b$ and $c$ such that $\omega$ is a solution of the following stochastic differential equation:
\begin{equation}\label{eq:masterSDE}
	\begin{split}
		\omega_t^{\tau_n } = \omega_0 & + \int_0^{t\wedge \tau_n} b_s \omega_s \sigma_s \, dZ_s + \int_0^{t \wedge \tau_n} c_s \omega_s \sigma_s \, dW_s \\
		& + \int_0^{t \wedge \tau_n} \bigg\{\bigg(\frac{b^2_s + c^2_s}{16}\bigg) \omega_s^2 + \bigg(\frac{b^2_s + c^2_s - 2 c_s}{4}\bigg) \omega_s - \bigg[\frac{b^2_s + c^2_s}{4} k^2_s + c_s k_s + 1\bigg]\bigg\} \sigma_s^2 \, ds, \quad a.s.
	\end{split}
\end{equation}
for any $t \in [0,+\infty)$. Eq.~\eqref{eq:masterSDE} is the combination of Eq.~\eqref{eq:SDE} and~\eqref{eq:nodrift} and we call it the \emph{IRV master SDE}.

We note that the no-drift condition in Eq.~\eqref{eq:nodrift} has the same advantages as the one derived by~\textcite{CarrSun:2014}, but now the IRV dynamics can accommodate more general drift and diffusion coefficients, because $a$, $b$, $c$ are processes, possibly dependent on $\omega$, $\sigma$, $T$, $K$ or further random factors.

Another critical novelty is the idea of localizing in time the IRV and call price processes. It is shared with~\textcite{ElAmraniJacquierMartini:2020}, and it allows to identify no-arbitrage conditions and embed smile bubbles into the new framework.

In this regard, if the sequences of stopping times in Eq.~\eqref{eq:stopping} almost surely converge, we adopt the notation: 
\begin{equation*}
	\tau^0 \coloneqq \lim_{n \rightarrow + \infty} \tau_n^0, \qquad \tau^{\infty} \coloneqq \lim_{n \rightarrow + \infty} \tau_n^{\infty}, \qquad \tau \coloneqq \tau^0 \wedge \tau^{\infty}.
\end{equation*}

Then:

\begin{proposition}\label{prop:noa1}
	Let $\omega$ be the market IRV process. If $\omega$ follows a locally consistent dynamics such that, almost surely:
	\begin{itemize}
		\item $\tau < T$, then the market is arbitrage-free on $[0,\tau)$; 
		\item $\tau = \tau^0 = T$, then the market is arbitrage-free and $\omega$ is said to follow a \emph{globally consistent dynamics};
		\item $\tau = \tau^{\infty} = T$ or $\tau > T$, then we almost surely observe a jump in the call price at time $T$.
	\end{itemize}
\end{proposition}

\begin{proof}
	Denote the call price process with $C$. Being $\omega$ the market IRV process, we have 
	\[
	C_t^{\tau_n} = S_t^{\tau_n} \cdot BS\Big(k_t^{\tau_n},\sqrt{\omega_t^{\tau_n}}\Big) \quad \Longrightarrow \quad C_t^{\tau_n} \in \Big(\big(S_t^{\tau_n} - K\big)_+, S_t^{\tau_n}\Big), \qquad  a.s.
	\]
	for any $n \ge 1$ and any $t \in [0,T)$. Moreover, the local consistency of the IRV dynamics guarantees that $C^{\tau_n}$ is a martingale, for any $n \ge 1$. If almost surely $\tau < T$, these two conditions ensure that the market is arbitrage-free on $[0,\tau)$, although the call price may not be equal to the conditional expectation of the final payoff under the risk-neutral probability.
	
	Denote the corresponding put price process with $P$. 
	As before, we have
	\[
	P_t^{\tau_n} = S_t^{\tau_n} \cdot {BS}_P\Big(k_t^{\tau_n},\sqrt{\omega_t^{\tau_n}}\Big) \quad \Longrightarrow \quad P_t^{\tau_n} \in \Big(\big(K - S_t^{\tau_n}\big)_+, K \Big), \qquad a.s.
	\]
	and now $P^{\tau_n}$ is a continuous uniformly bounded (by $K$) martingale, for any $n \ge 1$. Therefore, the Conditional Dominated Convergence Theorem applies and, if almost surely $\tau = \tau^0 = T$, we have
	\[
	\begin{split}
		P_t & = \lim_{n \rightarrow +\infty} P_t^{\tau_n} \\
		& = \lim_{n \rightarrow +\infty} \mathbb{E}^{\mathbb{Q}} \big[ P_{\tau_n}| \mathcal{F}_t\big] \\
		& = \lim_{n \rightarrow +\infty} \mathbb{E}^{\mathbb{Q}} \big[ S_{\tau_n} \cdot {BS}_P\big(k_{\tau_n},\sqrt{\omega_{\tau_n}}\big) \big| \mathcal{F}_t\big] \\
		& = \mathbb{E}^{\mathbb{Q}} \big[ S_{T} \cdot {BS}_P(k_{T},0) | \mathcal{F}_t\big] \\
		& =	\mathbb{E}^{\mathbb{Q}} \big[(K - S_{T} )_+| \mathcal{F}_t\big]
	\end{split}
	\]
	that is, by the Put-Call-Parity,
	\[
	C_t =	\mathbb{E}^{\mathbb{Q}} \big[(S_{T} - K)_+|\mathcal{F}_t\big], \qquad a.s.
	\]
	for any $t \in [0,T)$. In other words, the call price is equal to the conditional expectation of the final payoff under the risk-neutral probability, and thus the market is arbitrage-free.
	
	Finally, if almost surely $\tau = \tau^{\infty} = T$ or $\tau > T$, it is
	\[
	\lim_{t \rightarrow T^-} C_t = \lim_{t \rightarrow T^-} S_t \cdot BS\Big(k_t,\sqrt{\omega_t}\Big) = S_T \cdot BS\Big(k_T,\sqrt{\omega_T}\Big) > \big(S_T - K\big)_+, \qquad a.s.
	\]
	while
	\[
	C_T \equiv \big(S_T - K\big)_+, \qquad a.s.
	\]
	Hence, $C$ is a càdlàg process and we observe a jump at $T$ almost surely.
\end{proof}

Consistently with~\textcite{ElAmraniJacquierMartini:2020}, we set:
\begin{definition}
	A \emph{smile bubble} is an IRV process which follows a locally (but not globally) consistent dynamics.
\end{definition}

Let $\omega$ be an IRV process following a locally consistent dynamics such that $\tau < T$, almost surely. Proposition~\ref{prop:noa1} ensures that the market is arbitrage-free on $[0,\tau)$, meaning that no arbitrage opportunities exist as long as trades occur within the the bubble lifespan, with possibly additionally unwinding trades at the expiry $T$. As shown in the proof above, absence of arbitrage holds even though the call price may not be equal to the conditional expectation of the final payoff under the risk-neutral probability. At $\tau$, the IRV process either vanishes or explodes while the smile bubble blows up, giving rise to a static arbitrage opportunity and a change of regime. In particular, if $\tau = \tau^0$ and $S_{\tau} \leq K$, then $C_{\tau} = (S_{\tau} - K)_+ = 0$; instead, if $S_{\tau} > K$ and an investor opens a long position on the call option, a short position on the underlying asset and deposits \$$K$ on the checking account, then it is a static arbitrage opportunity. On the other hand,  if $\tau = \tau^{\infty}$, then $C_{\tau} = S_{\tau}$ and a static arbitrage opportunity is locked opening a short position on the call option and a long position on the underlying asset.

\subsection{Market with an IRV smile}

We now extend previous findings to a market with an IRV smile, that is a market on which are traded European call options with different strike prices $\{K_i\}_{i \in \mathcal{I}}$, but same maturity and underlying. We start by assuming $\mathcal{I} = \{1,\dots,N\}$ and $0 < K_1 < \cdots < K_N < +\infty$. To simplify the notation, we replace the argument $(K_i)$ with the subscript $i$, e.g., $C_{i} \coloneqq C(K_i)$.

By~\textcite{CarrMadan:2005}, we say that:
\begin{definition}\label{def:nostaticarb}
	The market is free of \emph{static arbitrage} at time $t \in [0,T)$ if, for any $i \in \mathcal{I}$, all the following conditions hold:
	\begin{enumerate}
		\item $\big(S_t - K_i \big)_+ < C_{i,t} < S_t \,$;
		\item $-1 < \frac{C_{i+1,t} - C_{i,t}}{K_{i+1}-K_i} < 0 \,$;
		\item $\frac{C_{i+1,t} - C_{i,t}}{K_{i+1} - K_i} \le \frac{C_{i+2,t} - C_{i+1,t}}{K_{i+2} - K_{i+1}} \,$;
		\item $-1 < \frac{C_{1,t} - S_t}{K_1} \le \frac{C_{2,t} - C_{1,t}}{K_{2} - K_1} < 0\,$.
	\end{enumerate}
\end{definition}

We recall that the concept of absence of static arbitrage is somewhat limited because it is equivalent to the impossibility of locking an arbitrage through a buy-and-hold position. However, it does not deal with the evolution in time of prices. 


Previous definitions adapt to the broadened environment as follows:

\begin{definition}
	On the filtered probability space $\big(\Omega,\mathcal{F},(\mathcal{F}_t)_{t \ge 0},\mathbb{Q}\big)$, let $\mathbi{w} = (\omega_{1},\dots,\omega_{N})$ be an $N$-dimensional process. Assume that each component follows a locally consistent dynamics. In particular, assume that, for any $i \in \mathcal{I}$ and $n \ge 1$, there exist progressively measurable processes $a_{i}$, $b_{i}$, $c_{i}$ such that
	\begin{equation*}
		\omega_{i,t}^{\tau_{i,n}} = \omega_{i,0} + \int_0^{t \wedge \tau_{i,n}} a_{i,s} \sigma_s^2 \, ds + \int_0^{t\wedge \tau_{i,n}} b_{i,s} \omega_{i,s} \sigma_s \, dZ_{i,s} + \int_0^{t \wedge \tau_{i,n}} c_{i,s} \omega_{i,s} \sigma_s \, dW_s, \qquad a.s.
	\end{equation*}
	where $a_{i}$ satisfies the no-drift condition of Eq.~\eqref{eq:nodrift} and $\{Z_{i}\}_{i \in \mathcal{I}}$ is a collection of standard Brownian motions independent of $W$. If the market $\big(S_0, \{{C}_{i,0}^{BS}\}_{i \in \mathcal{I}}\big)$ is free of static arbitrage, we say that $\mathbi{w}$ follows a \emph{locally consistent dynamics} and we write, for any $n \ge 1$,
	\[
	\mathbi{w}_t^{\tau_n} = \mathbi{w}_0 + \int_0^{t \wedge \tau_n} \mathbf{a}_s \sigma^2_s \, ds + \int_0^{t \wedge \tau_n} \sum_{i \in \mathcal{I}} b_{i,s} \omega_{i,s} \mathbf{e}_i \mathbf{e}_i' \, d\mathbi{Z}_s + \int_0^{t \wedge \tau_n} \sum_{i \in \mathcal{I}} c_{i,s} \omega_{i,s} \mathbf{e}_i \, dW_s, \qquad a.s. 
	\]
	where $\mathbf{a} = (a_{1},\dots,a_{N})$, $\mathbi{Z} = (Z_{1},\dots,Z_{N})$, $\{\mathbf{e}_i\}_{i \in \mathcal{I}}$ is the collection of unit vectors, and
	\begin{equation}\label{eq:taunsmile}
		\tau_n^0 \coloneqq \inf_{i \in \mathcal{I}} \big\{\tau_{i,n}^0\big\}, \qquad \tau_n^{\infty} \coloneqq \inf_{i \in \mathcal{I}} \big\{\tau_{i,n}^{\infty}\big\}, \qquad \tau_n \coloneqq \tau_n^0 \wedge \tau_n^{\infty}.
	\end{equation}
\end{definition}

\begin{definition}
	Let $\mathbi{w}$ follow a locally consistent dynamics. If each component $\{\omega_{i}\}_{i \in \mathcal{I}}$ follows a globally consistent dynamics, we say that $\mathbi{w}$ follows a \emph{globally consistent dynamics}.
\end{definition}

\begin{definition}
	We say that $\mathbi{w}$ is the market \emph{IRV smile process} if each component $\{\omega_{i}\}_{i \in \mathcal{I}}$ is the IRV process corresponding to the call with strike price $K_i$, i.e., $\mathbb{Q}\big\{ C^{\tau_{i,n}}_{i,t} = C^{BS,\tau_{i,n}}_{i,t}; \, t \in [0,T) \big\}$ for any $i \in \mathcal{I}$ and $n \ge 1$.
\end{definition}

Assuming the sequences of stopping times in Eq.~\eqref{eq:taunsmile} almost surely converge, we set
\begin{equation} \label{eq:tausmile}
	\tau^0 \coloneqq \lim_{n \rightarrow + \infty} \tau_{n}^0, \qquad \tau^{\infty} \coloneqq \lim_{n \rightarrow + \infty} \tau_{n}^{\infty}, \qquad \tau \coloneqq \tau^{\infty} \wedge \tau^0.
\end{equation}
Therefore, Proposition~\ref{prop:noa1} can be adapted as follows: 


\begin{proposition}\label{prop:noa2}
	If there exists an IRV smile process $\mathbi{w}$ following:
	\begin{itemize}
		\item a locally consistent dynamics such that $\tau < T$, almost surely, and the market is free of static arbitrage for any $t \in [0,\tau)$, then the market is arbitrage-free on $[0,\tau)$
		\item a globally consistent dynamics, then the market is arbitrage-free;
		\item a locally consistent dynamics such that $\tau = \tau^{\infty} = T$ or $\tau > T$, almost surely, then we observe a jump in at least one call price at time $T$.
	\end{itemize}
\end{proposition}

Before giving the proof, we note that when options with different strike prices are traded, the condition that the IRV smile follows a locally consistent dynamics is no longer sufficient to ensure absence of arbitrage, but we have to explicitly assume absence of static arbitrage. However, if the IRV smile dynamics is globally consistent, then the market is arbitrage-free, as in the single option case.

\begin{proof}
	If $\mathbi{w}$ follows a globally consistent dynamics, then each IRV process $\{\omega_{i}\}_{i \in \mathcal{I}}$ does the same. In particular, $\tau_i = \tau^0_i = T$, a.s., for any $i \in \mathcal{I}$. Therefore, from the proof of Proposition~\ref{prop:noa1}, we know that
	\[
	C_{i,t} = \mathbb{E}^{\mathbb{Q}}\big[(S_T-K_i)_+ | \mathcal{F}_t\big], \qquad a.s.
	\]
	for any $i \in \mathcal{I}$ and $t \in [0,T)$. Hence, the market is arbitrage-free.
	
	If instead $\mathbi{w}$ follows a locally consistent dynamics with $\tau < T$, it is almost surely $\tau \le \tau_i$, for any $i \in \mathcal{I}$. Then, for any $i \in \mathcal{I}$ and $t \in [0,+\infty)$, 
	\[
	C_{i,t}^{\tau} = S_t^{\tau} \cdot BS\Big(k_{i,t}^{\tau}, \sqrt{\omega_{i,t}^{\tau}}\Big), \qquad a.s.
	\]
	is a martingale and, for any $t \in [0,\tau)$, it satisfies condition 1 of Definition~\ref{def:nostaticarb}. Assuming conditions 2--4 also hold, the market is arbitrage-free on $[0,\tau)$, although option prices may not be equal to the expectation of the final payoffs under the risk-neutral probability.
	
	Finally, if $\tau = \tau^{\infty} = T$ or $\tau > T$, almost surely, then there exists at least one $i \in \mathcal{I}$ such that $\omega_{i,T} \neq 0$, so that
	\[
	\lim_{t \rightarrow T^-} C_{i,t} = \lim_{t \rightarrow T^-} S_t \cdot BS\big(k_{i,t}, \sqrt{\omega_{i,t}}\big) = S_T \cdot BS\Big(k_{i,T},\sqrt{\omega_{i,T}}\Big) > \big(S_T - K_i\big)_+, \qquad a.s.
	\]
	while
	\[
	C_{i,T} \equiv \big(S_T - K_i\big)_+, \qquad a.s.
	\]
	Hence, $C_i$ is a càdlàg process and we observe a jump at $T$ almost surely.
\end{proof}

From the proof, it immediately follows:
\begin{corollary}\label{cor:noa}
	If each IRV process $\{\omega_i\}_{i \in \mathcal{I}}$ follows a globally consistent dynamics, then the market is arbitrage-free.
\end{corollary}

Moreover, the proof stresses the criticality of the condition $\tau = \tau^0 = T$, almost surely. It distinguishes globally from locally consistent dynamics and ensures that call prices are equal to the conditional expectation of ﬁnal payoﬀs under the risk-neutral probability. This makes call prices consistent with the traditional risk-neutral valuation formula and thus with absence of arbitrage, both static and dynamic. This is not the case when the IRV smile dynamics is just locally consistent. Indeed, locally consistency ensures that stopped call prices are martingales, which excludes dynamic arbitrage until $\tau$, and satisfy condition 1 for absence of static arbitrage. While these two conditions are sufficient for a single option market to be arbitrage-free (until $\tau$), now we have to explicitly request that also conditions 2--4 of Definition~\ref{def:nostaticarb} hold. All this discussion is absent in~\textcite{CarrSun:2014} and represents the~\textcite{ElAmraniJacquierMartini:2020}'s transposition into the new framework.

All the findings of this section can be easily extended to a countable set of strike prices adopting the following definitions:
\begin{definition}
	Given a countable set $\mathcal{I}$  and a collection of strike prices $\{K_i\}_{i \in \mathcal{I}}$ such that $0 < K_i < K_{i+1} < +\infty$ for any $i  \in \mathcal{I}$, $\mathbi{w} = (\omega_{i})_{i \in \mathcal{I}}$ is said to follow a \emph{locally} [resp., \emph{globally}] \emph{consistent dynamics} if, for any finite subset $\iota \subseteq \mathcal{I}$, the dynamics of $\mathbi{w}^{\iota} = (\omega_{i})_{i \in \iota}$ is locally [resp., globally] consistent. Furthermore, $\mathbi{w}$ is the market \emph{IRV smile process} if, for any $\iota$, $\mathbi{w}^{\iota} = (\omega_{i})_{i \in \iota}$ is the IRV smile process of the restricted market.
\end{definition}

Finally, we note that the proof of Proposition~\ref{prop:noa2} applies without modifications to a market with a continuum of strike prices, so that:
\begin{corollary}\label{cor:noa2}
	Consider a market with a continuum of strike prices. For any $K \in (0,+\infty)$, let $\omega(K)$ be the corresponding market IRV process and $\tau^0$, $\tau^{\infty}$ and $\tau$ be defined as in Eq.~\eqref{eq:tausmile}. If each $\omega^{\tau}(K)$ follows:
	\begin{itemize}
		\item a locally consistent dynamics, $\tau < T$, a.s., and the market is free of static arbitrage for any $t \in [0,\tau)$, then the market is arbitrage-free on $[0,\tau)$; 
		\item a globally consistent dynamics, then the market is arbitrage-free;
		\item a locally consistent dynamics and $\tau = \tau^{\infty} = T$ or $\tau > T$, a.s., then we observe a jump in at least one call price at time $T$.
	\end{itemize}
\end{corollary}

%% file: Examples/Examples.tex
\section{Examples}
\label{sec:examples}


\subsection{Globally consistent dynamics in a single option market}

The results of~\textcite{SchweizerWissel:2008b} for a single option market can be mapped into the new IRV framework. Indeed, they identified a drift restriction on the implied variance process $\mathcal{X}$, which is nothing else than $\omega / (T-t)$. We omit the argument $(K)$ because fixed. 

According to the new framework, Eq.~(3.4)--(3.5) of~\textcite{SchweizerWissel:2008b} can be rewritten as\footnote{\textcite{SchweizerWissel:2008b} denote with $b$ the market price of risk. Under $\mathbb{Q}$, we have $b \equiv 0$ and $\mu \equiv 0$.}
\begin{equation}\label{eq:SW1}
	\begin{cases}
		dS_t = S_t 
		\begin{bmatrix}
			\sigma_{1,t} & 0
		\end{bmatrix}
		\times 
		\begin{bmatrix}
			dW_t \\
			dZ_t
		\end{bmatrix}, \qquad S_0 > 0 \\
		d\mathcal{X}_t = u_t \mathcal{X}_t \, dt + \mathcal{X}_t
		\begin{bmatrix}
			v_{1,t} & v_{2,t}
		\end{bmatrix}
		\times 
		\begin{bmatrix}
			dW_t \\
			dZ_t
		\end{bmatrix}, \qquad \mathcal{X}_0 > 0
	\end{cases}
\end{equation}
where $W$ and $Z$ are two independent Brownian motions, and the drift restriction (Eq.~(3.7) in their paper) as
\begin{multline}\label{eq:SW2}
	u_t = \frac{1}{T-t} \bigg( 1 - \frac{\sigma_{1,t}^2 + k_t \sigma_{1,t} v_{1,t} + \frac{1}{4} k_t^2 v_{1,t}^2 + \frac{1}{4} k_t^2 v_{2,t}^2}{\mathcal{X}_t} \bigg) \\
	+ \bigg( \frac{(T-t) \mathcal{X}_t}{16} + \frac{1}{4}\bigg) \big( v_{1,t}^2 + v_{2,t}^2\big) - \frac{1}{2}\sigma_{1,t} v_{1,t}
\end{multline}
where $k_t \coloneqq \log(K/S_t)$.

Assuming the processes $\mathbi{v}$ and $\mathbf{\sigma}$ to be functions of $t$, $S$ and $\mathcal{X}$,
\begin{gather*}
	\mathbi{v}_t = \mathbi{v}(t,S_t,\mathcal{X}_t) = 
	\begin{bmatrix}
		v_1(t,S_t,\mathcal{X}_t) \\
		v_2(t,S_t,\mathcal{X}_t)
	\end{bmatrix}, \\
	\mathbf{\sigma}_t = \mathbf{\sigma}(t,S_t,\mathcal{X}_t) = 
	\begin{bmatrix}
		\sigma_1(t,S_t,\mathcal{X}_t) \\
		\sigma_2(t,S_t,\mathcal{X}_t)
	\end{bmatrix} = 
	\begin{bmatrix}
		\sigma_1(t,S_t,\mathcal{X}_t) \\
		0
	\end{bmatrix},
\end{gather*} 
their Proposition 3.3 reads as:
\begin{proposition}\label{prop:SW33}
	Let
		\begin{equation}\label{eq:SW33}
		\begin{bmatrix}
			\sigma_1(t,S_t,\mathcal{X}_t) \\
			0
		\end{bmatrix} \coloneqq - \frac{1}{2} k_t \,
		\begin{bmatrix}
			v_1(t,S_t,\mathcal{X}_t) \\
			v_2(t,S_t,\mathcal{X}_t)
		\end{bmatrix} + \sqrt{\mathcal{X}_t} \cdot \bigg( \begin{bmatrix}
			f_1(t,S_t,\mathcal{X}_t) \\
			f_2(t,S_t,\mathcal{X}_t)
		\end{bmatrix} + (T-t) \begin{bmatrix}
			g_1(t,S_t,\mathcal{X}_t) \\
			g_2(t,S_t,\mathcal{X}_t)
		\end{bmatrix} \bigg)
	\end{equation}
	where $\mathbi{f}$, $\mathbi{g}$, $\mathbi{v}: [0,T] \times (0,\infty)^2 \rightarrow \mathbb{R}^2$ are locally Lipschitz\footnote{A function $\mathbi{f} : \Theta \times U \rightarrow \mathbb{R}^k$ is said \emph{locally Lipschitz} on $U$ if $\mathbi{f}(\cdot,x)$ is bounded for fixed $x \in U$ and if there exists a continuous function $C(\cdot,\cdot)$ on $U^2$ such that
		\[
		\big\lvert \mathbi{f}(\theta,x) - \mathbi{f}(\theta,x')\big\rvert \le C(x,x') \, \lvert x - x' \rvert, \qquad \forall x, x' \in U, \, \theta \in \Theta.
		\]} on $(0,\infty)^2$ and satisfy
	\begin{equation}\label{eq:SWcon}
		\begin{cases}
			f_1^2(t,S_t,\mathcal{X}_t) + f_2^2(t,S_t,\mathcal{X}_t) = 1 \\
			g_1^2(t,S_t,\mathcal{X}_t) + g_2^2(t,S_t,\mathcal{X}_t) \le const \\
			v_1^2(t,S_t,\mathcal{X}_t) + v_2^2(t,S_t,\mathcal{X}_t) \le \frac{const}{(1 + \sqrt{\mathcal{X}_t} + \lvert k_t \rvert)^2}
		\end{cases}
	\end{equation}
	where $const$ denotes a generic positive constant whose value can change from one line to the next. Then the system in Eq.~\eqref{eq:SW1}, where $u$ satisfies Eq.~\eqref{eq:SW2}, has a unique positive (nonexploding) solution $(S,\mathcal{X})$ on $[0,T]$.
\end{proposition}

Two observations follow from Proposition~\ref{prop:SW33}.

First, we note that
\[
\sigma_1(t,S_t,\mathcal{X}_t) \coloneqq -\frac{k_t}{2} v_1(t,S_t,\mathcal{X}_t) + \sqrt{\mathcal{X}_t} \big[ f_1(t,S_t,\mathcal{X}_t) + (T-t) \, g_1(t,S_t,\mathcal{X}_t) \big]
\]
may change sign, which implies that in general the spot volatility of the underlying asset is not bounded from below.

Second, the main question now is how to choose $\mathbi{f}$, $\mathbi{g}$ and $\mathbi{v}$ so that the constraints in Eq.~\eqref{eq:SWcon} are satisfied. For example, a valid family of $(\mathbi{f},\mathbi{g},\mathbi{v})$ is obtained setting
\[
f_2(t,S_t,\mathcal{X}_t) = g_2(t,S_t,\mathcal{X}_t) = v_2(t,S_t,\mathcal{X}_t)\equiv 0
\]
and specifying $f_1, g_1, v_1$ so that 
\[
\begin{cases}
	f_1(t,S_t,\mathcal{X}_t) \equiv \pm 1 \\
	\lvert g_1(t,S_t,\mathcal{X}_t) \rvert \le const \\
	\lvert v_1(t,S_t,\mathcal{X}_t) \rvert \le \frac{const}{1 + \sqrt{\mathcal{X}_t} + \lvert k_t \rvert}.
\end{cases}
\]
Starting from here, in the next subsection, we provide an explicit sub-family of valid $(\mathbi{f},\mathbi{g},\mathbi{v})$.


\subsubsection{An explicit sub-family}

Under the assumption
\[
f_2(t,S_t,\mathcal{X}_t) = g_2(t,S_t,\mathcal{X}_t) = v_2(t,S_t,\mathcal{X}_t)\equiv 0,
\]
let
\begin{gather*}
	f_1(t,S_t,\mathcal{X}_t) \equiv 1, \\
	g_1(t,S_t,\mathcal{X}_t) \text{ be a bounded and locally Lipschitz on $(0,\infty)^2$ function}, \\
	v_1(t,S_t,\mathcal{X}_t) \coloneqq \frac{S_t}{1 + \mathcal{X}_t + S_t^2} \, w_1(t,S_t,\mathcal{X}_t),
\end{gather*}
where $w_1(t,S_t,\mathcal{X}_t)$ is a bounded and locally Lipschitz on $(0,\infty)^2$ function.

It follows that all the constraints in Eq.~\eqref{eq:SWcon} are satisfied and, defining 
\[
\sigma_1(t,S_t,\mathcal{X}_t) \coloneqq - \frac{k_t}{2} \frac{S_t}{1 + \mathcal{X}_t + S_t^2} w_1(t,S_t,\mathcal{X}_t) + \sqrt{\mathcal{X}_t} \big[ 1 + (T-t) g_1(t,S_t,\mathcal{X}_t)\big],
\]
Proposition~\ref{prop:SW33} ensures the existence of well-behaved solutions to the system in Eq.~\eqref{eq:SW1}, where $u$ satisfies Eq.~\eqref{eq:SW2}. The family of solutions is indexed by the double family of bounded and locally Lipschitz functions $(g_1, w_1)$.

For example, we can take  
\[
\begin{cases}
	g_1(t,S_t,\mathcal{X}_t) \equiv 1 \\
	w_1(t,S_t,\mathcal{X}_t) = - \sgn(k_t)
\end{cases}
\]
so that $\sigma_1 > 0$.

Another choice ensuring the strictly positivity of the spot volatility is 
\[
\begin{cases}
	g_1(t,S_t,\mathcal{X}_t) \geq 0 \\
	w_1(t,S_t,\mathcal{X}_t) \equiv 0.
\end{cases}
\]
In this case, $\mathcal{X}_t$ is a finite variation process satisfying
\[
d\mathcal{X}_t = u_t \mathcal{X}_t \, dt
\]
where
\[
u_t = \frac{1}{T-t} \bigg( 1 - \frac{\sigma_1^2(t,S_t,\mathcal{X}_t)}{\mathcal{X}_t} \bigg).
\]
Moreover, it holds:
\[
\begin{split}
	d(T-t) \mathcal{X}_t & = - \mathcal{X}_t \, dt + (T-t) d\mathcal{X}_t \\
	& = - \mathcal{X}_t \, dt + (T-t) u_t \mathcal{X}_t dt \\
	& = - \sigma_1^2(t,S_t,\mathcal{X}_t) \, dt
\end{split}
\]
where
\[
\sigma_1(t,S_t,\mathcal{X}_t) = \sqrt{\mathcal{X}_t} \big[ 1 + (T-t) \, g_1(t,S_t,\mathcal{X}_t)\big].
\]

\subsubsection{Mapping into the IRV framework}

Starting from Eq.~\eqref{eq:SW1}--\eqref{eq:SW2}, we derive the IRV dynamics under the assumption $\sigma_1 \ne 0$:
\[
\begin{split}
	d\omega_t = & \, d\big((T-t) \mathcal{X}_t \big) \\
	= & \, -\mathcal{X}_t \, dt + (T-t) \, d\mathcal{X}_t \\
	= & \, \Bigg[ \frac{v_{2,t}^2 / \sigma_{1,t}^2 + v_{1,t}^2 / \sigma_{1,t}^2}{16} \, \omega_t^2 + \frac{v_{2,t}^2 / \sigma_{1,t}^2 + v_{1,t}^2 / \sigma_{1,t}^2 - 2 v_{1,t} / \sigma_{1,t}}{4} \, \omega_t \\
	& - \bigg( \frac{v_{2,t}^2 / \sigma_{1,t}^2 + v_{1,t}^2 / \sigma_{1,t}^2}{4} k_t^2 + \frac{v_{1,t}}{\sigma_{1,t}} k_t + 1 \bigg) \Bigg] \sigma_{1,t}^2 \, dt + v_{2,t} \omega_t \, dZ_t + v_{1,t} \omega_t \, dW_t.
\end{split}
\]
Defining
\begin{equation*}
	\begin{cases}
		b_t \coloneqq v_{2,t} / \sigma_{1,t} \\
		c_t \coloneqq v_{1,t} / \sigma_{1,t}
	\end{cases}
\end{equation*}
we obtain
\[
d\omega_t = \Bigg[ \frac{b_t^2 + c_t^2}{16} \, \omega_t^2 + \frac{b_t^2 + c_t^2 - 2 c_t}{4} \, \omega_t - \bigg( \frac{b_t^2 + c_t^2}{4} k_t^2 + c_t k_t + 1 \bigg) \Bigg] \sigma_{1,t}^2 \, dt + b_t \sigma_{1,t} \omega_t \, dZ_t + c_t \sigma_{1,t} \omega_t \, dW_t
\]
which is exactly the IRV master SDE of Eq.~\eqref{eq:masterSDE}. It means that if $\mathcal{X}$ satisfies the drift restriction derived by~\textcite{SchweizerWissel:2008b}, then the IRV satisfies the no-drift condition in Eq.~\eqref{eq:nodrift} and follows a locally consistent dynamics. 

Being $\omega_t = (T-t) \mathcal{X}_t$, it immediately follows:
\begin{corollary}\label{cor:SWIRV}
	If the conditions in Proposition~\ref{prop:SW33} are satisfied, then the IRV process exists as the unique solution of the IRV master SDE and follows a globally consistent dynamics.
\end{corollary}

Corollary~\ref{cor:SWIRV} requires to specify an expression for the spot volatility, contrary to the aim of the new framework of identifying arbitrage-free IRV models leaving the spot volatility dynamics unspecified. However, it is relevant because it ensures the existence of well-defined IRV processes following a globally consistent dynamics. For example, those corresponding to the previous explicit sub-family. In those cases, we have
\[
\begin{cases}
	dS_t = S_t \sigma_1(t,S_t,\mathcal{X}_t) \, dW_t \\
	d\omega_t = \Big[ \frac{c_t^2}{16} \, \omega_t^2 + \frac{c_t^2 - 2 c_t}{4} \, \omega_t - \Big( \frac{c_t^2}{4} k_t^2 + c_t k_t + 1 \Big) \Big] \sigma_1^2(t,S_t,\mathcal{X}_t) \, dt + c_t \sigma_1(t,S_t,\mathcal{X}_t)\,  \omega_t \, dW_t
\end{cases}
\]
where
\[
\sigma_1(t,S_t,\mathcal{X}_t) = -\frac{k_t}{2} v_1(t,S_t,\mathcal{X}_t) + \sqrt{\mathcal{X}_t} \big[ f_1(t,S_t,\mathcal{X}_t) + (T-t) \, g_1(t,S_t,\mathcal{X}_t) \big]
\]
and
\[
c_t = \hat{c}(t,S_t,\mathcal{X}_t) = \frac{v_1(t,S_t,\mathcal{X}_t)}{\sigma_1(t,S_t,\mathcal{X}_t)}.
\]
Here, we note that the Brownian motion driving the IRV is just $W$, the same driving the underlying asset price. Moreover, $c$ tends to infinite as $\sigma_1$ tends to $0$, whereas the product $c \sigma_1$ remain bounded.

Finally, we recall that~\textcite{SchweizerWissel:2008b} addressed the problem of absence of arbitrage in a market with a smile, too. In that case, they worked with what they named \emph{local implied volatility} and \emph{price level}, identifying different drift restrictions and conditions for the existence of well-behaved solutions. However, the map of those results in terms of implied remaining variance and underlying asset price requires further investigation.


\subsection{Surface Stochastic Volatility Inspired parametrization}

The main question of~\textcite{ElAmraniJacquierMartini:2020} was whether there exist IRV smile models such that  smiles at any time satisfy the symmetric \emph{Surface Stochastic Volatility Inspired (SSVI)} parametrization (see~\cite{Gatheral:2004},~\cite{Gatheral:2006},~\cite{GatheralJacquier:2014}) and the market is arbitrage-free. 

Fix a maturity $T$ and, for any $K \in (0,+\infty)$, consider the IRV smile model
\begin{equation}\label{eq:SVI1}
\begin{cases}
	S_t = \exp\Big\{\int_0^t \sigma_s \, dW_s - \frac{1}{2} \int_0^t \sigma_s^2 \, ds\Big\} \\
	\omega_t(K) = \frac{1}{2} \Big( \theta_t + \sqrt{\theta_t^2 + \psi_T^2 k_t^2}\Big) 
\end{cases}
\end{equation}
where $W$ is a standard Brownian motion, $\sigma$ is an almost surely positive process satisfying the Novikov condition\footnote{$\mathbb{E}\Big[\exp\Big\{\frac{1}{2} \int_0^T \sigma_t^2 \, dt\Big\}\Big] < \infty$.}, $k_t \coloneqq \log(K/S_t)$, $\psi_T$ is a nonnegative constant and, according with the SSVI parametrization,  $\theta$ is the IRV process corresponding to the time-$t$ at-the-money call. Moreover, assume that
\begin{equation}\label{eq:SVI2}
	d\theta_t = \frac{(\psi_T -4)(\psi_T + 4)}{16} \sigma_t^2 \, dt - \psi_T \sigma_t \, dZ_t \\
\end{equation}
where $Z$ is a standard Brownian motion independent of $W$. We note that, a priori, Eq.~\eqref{eq:SVI2} does not guarantee that $\theta$ is positive almost surely.

\textcite{ElAmraniJacquierMartini:2020} proved that the market described by Eq.~\eqref{eq:SVI1} is arbitrage-free only if Eq.~\eqref{eq:SVI2} is satisfied, $\theta$ almost surely converges to $0$ as $t \rightarrow T$ and $\psi_T = 0$. It follows that IRV smiles are always flat. However, if the market described by Eq.~\eqref{eq:SVI1}--\eqref{eq:SVI2} is free of static arbitrage at the initial time, then they proved that it is also arbitrage-free up to a stopping time $\tau < T$, a.s. In particular, $\tau$ is the first time when static arbitrage conditions are violated. This is what they named a \emph{smile bubble}. Consistently with Section~\ref{sec:masterSDE}, we aim to show that SSVI smile bubbles correspond to locally consistent dynamics, i.e., $\omega$ derived from Eq.~\eqref{eq:SVI1}--\eqref{eq:SVI2} satisfies the IRV master SDE in Eq.~\eqref{eq:masterSDE} for any $K$.

For notational simplicity, we set 
\[
A_t(K) \coloneqq \sqrt{\theta_t^2 + \psi_T^2 k_t^2}
\]
where $\theta$ satisfies Eq.~\eqref{eq:SVI2}, so that
\[
\omega_t(K) = \frac{1}{2} \big( \theta_t + A_t(K) \big).
\]

For $n \ge 1$, we define the following sequences of stopping times
\[
\xi_n^0 \coloneqq \inf \big\{ t \ge 0: \theta_t  \leq 1/n \big\}, \qquad \xi_n^{\infty} \coloneqq \inf \big\{ t \ge 0: \theta_t \ge n \big\}, \qquad \xi_n \coloneqq \xi_n^0 \wedge \xi_n^{\infty},
\]
and we assume that they almost surely converge
\[
\xi^0 \coloneqq \lim_{n \rightarrow + \infty} \xi_n^0, \qquad \xi^{\infty} \coloneqq \lim_{n \rightarrow + \infty} \xi_n^{\infty}, \qquad \xi \coloneqq \xi^0 \wedge \xi^{\infty}.
\]
In this way, $A^{\xi} _t(K) \coloneqq A_{t \wedge \xi}(K) \in (0,+\infty)$ and $\omega^{\xi}_t(K) \coloneqq \omega_{t \wedge \xi}(K) \in (0,+\infty)$, for any $K \in (0,+\infty)$ and $t \in [0,T)$.

We apply It\^o formula to $\omega^{\xi}$, omitting the argument $(K)$:
\begin{equation}\label{eq:SVI3}
	\begin{split}
		\omega_t^{\xi} = \omega_0 & \, + \int_0^{t \wedge \xi} \frac{1}{2} \Bigg[ \bigg( 1 + \frac{\theta_s}{A_s}\bigg) \frac{(\psi_T - 4)(\psi_T + 4)}{16} + \frac{\psi_T^2 k_s}{2 A_s} + \frac{1}{2 A_s^3} \bigg( \psi_T^4 k_s^2 + \theta_s^2 \psi_T^2 \bigg)  \Bigg] \sigma_s^2 \, ds \\
		& \, - \int_0^{t \wedge \xi} \frac{1}{2} \bigg( 1 + \frac{\theta_s}{A_s} \bigg) \psi_T \sigma_s  \, dZ_s - \int_0^{t \wedge \xi} \frac{\psi_T^2 k_s}{2 A_s} \sigma_s \, dW_s.
	\end{split}
\end{equation}
We note that the framework is that of a market with an IRV smile driven by two independent 1-dimensional Brownian motions: $W$, the one driving the underlying asset price; and $Z$, common to all strike prices. 

To verify that Eq.~\eqref{eq:SVI3} is consistent with the IRV master SDE, let us assume $\omega_0 \in (0,+\infty)$ for any $K$, and denote
\begin{gather*}
	a_t^{\xi} = \frac{1}{2} \Bigg[ \bigg( 1 + \frac{\theta_t^{\xi}}{A_t^{\xi}}\bigg) \frac{(\psi_T -4)(\psi_T + 4)}{16} + \frac{\psi_T^2 k_t^{\xi}}{2 A_t^{\xi}} + \frac{1}{2 (A_t^{\xi})^3} \bigg( \psi_T^4 (k_t^{\xi})^2 + (\theta_t^{\xi})^2 \psi_T^2 \bigg)  \Bigg] \\
	b_t^{\xi} = - \frac{1}{2 \omega_t^{\xi}} \bigg( 1 + \frac{\theta_t^{\xi}}{A_t^{\xi}} \bigg) \psi_T \\
	c_t^{\xi} = - \frac{\psi_T^2 k_t^{\xi}}{2 A_t^{\xi} \omega_t^{\xi}}.
\end{gather*}
Then the no-drift condition of Eq.~\eqref{eq:nodrift} requires
\[
\begin{split}
	a_t^{\xi} = & \, \frac{1}{64} \Bigg[ \frac{\big( A_t^{\xi} + \theta_t^{\xi} \big)^2}{(A_t^{\xi})^2} \psi_T^2 + \frac{\psi_T^4 (k_t^{\xi})^2}{(A_t^{\xi})^2}\Bigg] + \frac{1}{16 \omega_t^{\xi}} \Bigg[ \frac{\big( A_t^{\xi} + \theta_t^{\xi}\big)^2}{(A_t^{\xi})^2} \psi_T^2+ \frac{\psi_T^4 (k_t^{\xi})^2}{(A_t^{\xi})^2} \Bigg] + \frac{\psi_T^2 k_t^{\xi}}{4 A_t^{\xi}}\\
	& - \frac{(k_t^{\xi})^2}{16 (\omega_t^{\xi})^2} \Bigg[\frac{\big( A_t^{\xi} + \theta_t^{\xi} \big)^2}{(A_t^{\xi})^2} \psi_T^2+ \frac{\psi_T^4 (k_t^{\xi})^2}{(A_t^{\xi})^2} \Bigg] + \frac{\psi_T^2 (k_t^{\xi})^2}{2 A_t^{\xi} \omega_t^{\xi}} - 1 \\
	= & \, \frac{1}{32 A_t^{\xi} (\omega_t^{\xi})^2}  \bigg[\psi_T^2 \big(A_t^{\xi} + \theta_t^{\xi} \big) (\omega_t^{\xi})^2 + 4 \psi_T^2 \big(A_t^{\xi} + \theta_t^{\xi} \big) \omega_t^{\xi} + 8 \psi_T^2 k_t^{\xi} (\omega_t^{\xi})^2 \\
	& \, - 4 \psi_T^2 (k_t^{\xi})^2 \big(A_t^{\xi} + \theta_t^{\xi} \big) + 16 \psi_T^2 (k_t^{\xi})^2 \omega_t^{\xi} - 32 A_t^{\xi} (\omega_t^{\xi})^2\bigg]
\end{split}
\]
that is
\begin{multline*}
	\frac{1}{2} \Bigg[ \bigg( 1 + \frac{\theta_t^{\xi}}{A_t^{\xi}}\bigg) \frac{(\psi_T -4)(\psi_T + 4)}{16} + \frac{\psi_T^2 k_t^{\xi}}{2 A_t^{\xi}} + \frac{1}{2 (A_t^{\xi})^3} \bigg( \psi_T^4 (k_t^{\xi})^2 + (\theta_t^{\xi})^2 \psi_T^2 \bigg)  \Bigg] = \\
	= \frac{1}{32 A_t^{\xi} (\omega_t^{\xi})^2}  \bigg[\psi_T^2 \big(A_t^{\xi} + \theta_t^{\xi} \big) (\omega_t^{\xi})^2 + 4 \psi_T^2 \big(A_t^{\xi} + \theta_t^{\xi} \big) \omega_t^{\xi} + 8 \psi_T^2 k_t^{\xi} (\omega_t^{\xi})^2 \\
	- 4 \psi_T^2 (k_t^{\xi})^2 \big(A_t^{\xi} + \theta_t^{\xi} \big) + 16 \psi_T^2 (k_t^{\xi})^2 \omega_t^{\xi} - 32 A_t^{\xi} (\omega_t^{\xi})^2\bigg]
\end{multline*}
equivalent to
\begin{multline*}
	\frac{\psi_T^2 \big(A_t^{\xi} + \theta_t^{\xi} + 8 k_t^{\xi} + 8\big) - 16 \big(A_t^{\xi} + \theta_t^{\xi} \big)}{32 A_t^{\xi}} = \\
	= \frac{1}{32 A_t^{\xi} (\omega_t^{\xi})^2}  \bigg[\psi_T^2 \big(A_t^{\xi} + \theta_t^{\xi} \big) (\omega_t^{\xi})^2 + 4 \psi_T^2 \big(A_t^{\xi} + \theta_t^{\xi} \big) \omega_t^{\xi} + 8 \psi_T^2 k_t^{\xi} (\omega_t^{\xi})^2 \\
	- 4 \psi_T^2 (k_t^{\xi})^2 \big(A_t^{\xi} + \theta_t^{\xi} \big) + 16 \psi_T^2 (k_t^{\xi})^2 \omega_t^{\xi} - 32 A_t^{\xi} (\omega_t^{\xi})^2\bigg]
\end{multline*}
and hence
\begin{multline*}
	\psi_T^2 \big(A_t^{\xi} + \theta_t^{\xi} + 8 k_t^{\xi} + 8\big) \big(A_t^{\xi} + \theta_t^{\xi} \big)^2 - 16 \big(A_t^{\xi} + \theta_t^{\xi} \big)^3 = \\
	= \psi_T^2 \big(A_t^{\xi} + \theta_t^{\xi} + 8 k_t^{\xi} + 8\big) \big(A_t^{\xi} + \theta_t^{\xi} \big)^2 - 16 \big(A_t^{\xi} + \theta_t^{\xi} \big)^3
\end{multline*}
which holds true.

We conclude that $\omega^{\xi}$ as derived from Eq.~\eqref{eq:SVI1}--\eqref{eq:SVI2} satisfies the IRV master SDE for any $K \in (0,+\infty)$ or, equivalently, that the corresponding IRV dynamics is locally consistent. Moreover, we note that $\xi$ controls for zeroing and explosiveness of $\omega$. Therefore, Corollary~\ref{cor:noa2} ensures that if, almost surely:
\begin{itemize}
	\item $\xi = \xi^0 = T$ (as it is the case when $\theta$ is almost surely finite, $\lim_{t \rightarrow T^-} \theta_t = 0$ and $\varphi_T = 0$), then the market is arbitrage-free; 
	\item the market is free of static arbitrage until $\xi$ and $\xi < T$, then the market is arbitrage-free on $[0,\xi)$.
\end{itemize} 

Conveniently,~\textcite{GatheralJacquier:2014} identified necessary and sufficient conditions on the SSVI parameters for absence of static arbitrage. From those,~\textcite{ElAmraniJacquierMartini:2020} derived that there is no static arbitrage if and only if $\psi_T^2 \le \mathfrak{B}(\theta)$ for any $\theta > 0$, where $\mathfrak{B}:[0,+\infty) \rightarrow [0,16]$ is defined as
\[
\mathfrak{B}(\theta) \coloneqq A(\theta) \, \mathds{1}_{\{\theta < 4\}} + 16 \, \mathds{1}_{\{\theta \ge 4\}}
\]
with
\[
A(\theta) \coloneqq \frac{16 \theta \zeta_{\theta} (\zeta_{\theta}+1)}{8(\zeta_{\theta}-2)+\theta \zeta_{\theta}(\zeta_{\theta}-1)} \qquad \text{and} \qquad \zeta_{\theta} \coloneqq \frac{2}{1-\theta/4} + \sqrt{\bigg(\frac{2}{1-\theta/4}\bigg)^2 + \frac{2}{1-\theta/4}}.
\]

Adopting the notation of~\textcite{ElAmraniJacquierMartini:2020}, let $\mathfrak{B}^{\leftarrow}: [0,16) \rightarrow [0,4)$ be the inverse of the restriction of the function $\mathfrak{B}$ to the domain $[0,4)$. Assuming $\psi_T \in [0,4)$, we can define the stopping time
\begin{equation}\label{eq:SVI4}
	\tau \coloneqq \xi \wedge \inf \big\{ t \ge 0: \theta_t < \mathfrak{B}^{\leftarrow}(\psi_T^2)\big\}.
\end{equation}

From the discussion above, it follows that:
\begin{itemize}
	\item if $\psi_T \in [0,4)$, $\theta_0 > \mathfrak{B}^{\leftarrow}(\psi_T^2)$, and $\tau < T$, almost surely, then the IRV smile dynamics implied by the SSVI parametrization in Eq.~\eqref{eq:SVI1}--\eqref{eq:SVI2} is locally consistent and the market is arbitrage-free on $[0,\tau)$;
	\item if $\theta$ is almost surely finite and converges to $0$ as $t \rightarrow T$, and $\psi_T = 0$, then each IRV dynamics is globally consistent, the market is arbitrage-free, but $\omega = \theta$ for any $t$ and $K$, i.e., we always observe flat smiles.
\end{itemize}
In other words, we have proved that the findings of~\textcite{ElAmraniJacquierMartini:2020} are consistent with those in Section~\ref{sec:masterSDE}, and SSVI smile bubbles are nothing else than IRV processes following a locally (but not globally) consistent dynamics.

%% file: Sandwiched_martingales/Sandwiched_martingales.tex
\section{Sandwiched martingales}
\label{sec:sandwiched}

In this section, we introduce the notion of sandwiched martingale and we show how it is related to that of locally consistent dynamics. We consider a single option market and a fixed maturity $T$, so that we can omit the arguments $(K,T)$ of $C$ and $\omega$ because fixed. 

Let $BS_{\bar{k}}: [0,+\infty] \rightarrow \big[ (1 - e^{\bar{k}})_+, 1\big]$ be the restriction of $BS$ to the second variable:
\[
BS_{\bar{k}}(v) \coloneqq BS\big(\bar{k},v\big), \qquad \bar{k} \in \mathbb{R}.
\]
The function $BS_{\bar{k}}$ is strictly increasing so that ${BS}_{\bar{k}}^{-1}: \big[ (1 - e^{\bar{k}})_+, 1\big] \rightarrow [0,+\infty]$ is well-defined.

From the definition of Implied Remaining Variance, it is
\[
\sqrt{\omega_t} = BS^{-1}_{\log(K/S_t)}\big(C^{obs}_t / S_t\big) = BS^{-1}_{k_t}\big(C^{obs}_t / S_t\big), \qquad \text{for any $t \in [0,T)$}.
\]
where $k_t \coloneqq \log(K/S_t)$. 

\begin{theorem}\label{thr:sandwichedres}
	Let $\big(\Omega, \mathcal{F}, (\mathcal{F}_t)_{t \ge 0},\mathbb{Q}\big)$ be a filtered probability space accommodating two independent Brownian motions $W$ and $Z$, and $(\mathcal{F}_t)_{t \ge 0}$ be the augmented filtration generated by them. Moreover:
	\begin{itemize}
		\item $\sigma \in L^2_{loc}\big((0,+\infty)\big)$ is such that the process $S$, defined by 
		\[
		S_t \coloneqq S_0 \cdot \exp\bigg\{ \int_0^t \sigma_s \, dW_s - \frac{1}{2} \int_0^t \sigma_s^2 \, ds \bigg\}, \qquad S_0 \in (0,+\infty), \qquad a.s.
		\]
		for any $t \in [0,+\infty)$, is a martingale. 
		\item $C$ is a progressively measurable process with continuous paths and such that
		\[
		\begin{cases}
			(S_0-K)_+ < C_0 < S_0 \\
			(S_t-K)_+ \le C_t \le S_t
		\end{cases} \qquad a.s.
		\]
		\item For $n \ge 1$, define the following sequences of stopping times
		\[
		\tau_n^0 \coloneqq \inf \Big\{t \ge 0: C_t \le (S_t-K)_+ + 1/n \Big\} \qquad \text{and} \qquad \tau_n^{\infty} \coloneqq \inf \Big\{t \ge 0: C_t \ge S_t - 1/n \Big\}
		\]
		and set $\tau_n \coloneqq \tau^0_n \wedge \tau_n^{\infty}$.
	\end{itemize}
	
	Then the process $C^{\tau_n}$ is a martingale and $\langle C^{\tau_n} \rangle$ is an almost surely absolutely continuous function with respect to time if and only if 
	\[
	\omega^{\tau_n}_t \coloneqq \Big( {BS}_{k_t^{\tau_n}}^{-1} \big(C^{\tau_n}_t/S^{\tau_n}_t\big)\Big)^2
	\]
	where $k_t \coloneqq \log(K/S_t)$, follows a locally consistent dynamics.
\end{theorem}

\begin{proof}
	For any $n \ge 1$, let $C^{\tau_n}$ be a martingale and $\langle C^{\tau_n} \rangle$ an almost surely absolutely continuous function with respect to time. By the martingale representation theorem, there exists a measurable adapted process $\nu$ such that 
	\[
	\int_0^{t \wedge \tau_n} \nu_s^2 \, ds < \infty \qquad \text{and} \qquad C^{\tau_n}_t = C_0 + \int_0^{t \wedge \tau_n} \nu_s \, dZ_s, \qquad a.s. \\
	\]
	Monotonicity of $BS_{k^{\tau_n}}^{-1}$ ensures that $\omega^{\tau_n} \in (0,+\infty)$, almost surely. Applying It\^o formula to $\omega^{\tau_n}$, we have
	\[
	\begin{split}
		\omega^{\tau_n}_t - \omega_0 = & \, 2 \int_0^{t \wedge \tau_n} \sqrt{\omega^{\tau_n}_s} \frac{\partial BS_{k_s}^{-1}(C_s/S_s)}{\partial C} \, dC_s \\
		& + \int_0^{t \wedge \tau_n} \bigg[ \bigg(\frac{\partial BS_{k_s}^{-1}(C_s/S_s)}{\partial C} \bigg)^2 + \sqrt{\omega_s} \frac{\partial^2 BS_{k_s}^{-1}(C_s/S_s)}{\partial C^2}\bigg] d\langle C \rangle_s \\
		= & \, 2 \int_0^{t \wedge \tau_n} \frac{\sqrt{\omega_s}}{S_s} \bigg(\frac{\partial BS_{k_s}(\sqrt{\omega_s})}{\partial \sqrt{\omega_s}}\bigg)^{-1} dC_s \\
		& + \int_0^{t \wedge \tau_n} \Bigg[ \frac{1}{S_s^2} \bigg( \frac{\partial BS_{k_s}(\sqrt{\omega_s})}{\partial \sqrt{\omega_s}} \bigg)^{-2} + \frac{\sqrt{\omega_s}}{S_s} \cdot \frac{\partial}{\partial C} \bigg(\bigg(\frac{\partial BS_{k_s}(\sqrt{\omega_s})}{\partial \sqrt{\omega_s}} \bigg)^{-1} \bigg) \Bigg] d \langle C \rangle_s \\
		= & \, 2 \int_0^{t \wedge \tau_n} \frac{\sqrt{\omega_s}}{S_s \cdot \varphi\big(d_+(k_s,\sqrt{\omega_s})\big)} \, dC_s \\
		& + \int_0^{t \wedge \tau_n} \bigg[ \frac{1}{S_s^2 \cdot \varphi^2\big(d_+(k_s,\sqrt{\omega_s})\big)} + \frac{\sqrt{\omega_s}}{S_s} \cdot \frac{\partial}{\partial C} \bigg(\frac{1}{\varphi\big(d_+\big(k_s,BS^{-1}_{k_s}(C_s/S_s)\big)\big)}\bigg) \bigg] d\langle C \rangle_s \\
		= & \, \int_0^{t \wedge \tau_n} \bigg( \frac{\omega_s^2}{4} + \omega_s - k_s^2 \bigg) \frac{\nu_s^2}{\omega_s \, S_s^2 \cdot \varphi^2\big(d_+(k_s,\sqrt{\omega_s})\big)} \, dt \\
		& + 2 \int_0^{t \wedge \tau_n} \frac{\sqrt{\omega_s} \cdot \nu_s}{S_s \cdot \varphi\big(d_+(k_s,\sqrt{\omega_s})\big)} \, dZ_s, \qquad a.s.
	\end{split}
	\]
	where $\varphi(\cdot)$ is the Gaussian probability density function. We observe that the no-drift condition in Eq.~\eqref{eq:nodrift}, if derived replacing the function $BS$ with its restriction $BS_k$,  reads as
	\[
	a^{\tau_n}_t = \bigg( \frac{(\omega_t^{\tau_n})^2}{4} + \omega_t^{\tau_n} - (k_t^{\tau_n})^2 \bigg)  \frac{(b_t^{\tau_n})^2 + (c_t^{\tau_n})^2}{4}, \qquad a.s.
	\]
	Hence, we can say that $\omega^{\tau_n}$ follows a locally consistent dynamics with
	\[
	b^{\tau_n}_n = \frac{2 \nu^{\tau_n}_t }{\sqrt{\omega_t^{\tau_n}} \, \sigma_t^{\tau_n} \, S_t^{\tau_n} \cdot \varphi\big(d_+\big(k_t^{\tau_n}, \sqrt{\omega_t^{\tau_n}}\big)\big)} \qquad \text{and} \qquad c_t^{\tau_n} \equiv 0, \qquad a.s.
	\]
	
	On the other hand,  let $\omega^{\tau_n} \coloneqq \big( {BS}_{k_t^{\tau_n}}^{-1} \big(C^{\tau_n}_t/S^{\tau_n}_t\big)\big)^2$ follow a locally consistent dynamics, that is, for any $n \ge 1$, there exist progressively measurable processes $a$, $b$, $c$ such that
	\[
	\int_0^{t \wedge \tau_n} \Big\{ \big\lvert a_s\sigma_s^2 \big\rvert + \big\lvert b_s \omega_s \sigma_s \big\rvert^2 + \big\lvert c_s \omega_s \sigma_s \big\rvert^2 \Big\} \, ds < \infty, \qquad \text{a.s.}
	\]
	and
	\[
	\omega_t^{\tau_n} = \omega_0 + \int_0^{t \wedge \tau_n} a_s \sigma_s^2 \, ds + \int_0^{t\wedge \tau_n} b_s \omega_s \sigma_s \, dZ_s + \int_0^{t \wedge \tau_n} c_s \omega_s \sigma_s \, dW_s, \qquad \text{a.s.}
	\]
	where
	\[
	a^{\tau_n}_t = \bigg( \frac{(\omega_t^{\tau_n})^2}{4} + \omega_t^{\tau_n} - (k_t^{\tau_n})^2 \bigg)  \frac{(b_t^{\tau_n})^2 + (c_t^{\tau_n})^2}{4}, \qquad a.s.
	\]
	for any $t \in [0,+\infty)$. Omitting the arguments $\big(k^{\tau_n}, \sqrt{\omega^{\tau_n}})$ of $d_-$, we apply the It\^o formula to $C^{\tau_n} = S^{\tau_n} \cdot BS_{k^{\tau_n}}\big(\sqrt{\omega^{\tau_n}}\big)$ and we obtain:
	\[
	\begin{split}
		C^{\tau_n}_t = & \, C_0 + \int_0^{t \wedge \tau_n} b_s \frac{K \sqrt{\omega_s} \, \varphi(d_-)}{2} \sigma_s \, dZ_s + \int_0^{t \wedge \tau_n} c_s \frac{ K \sqrt{\omega_s} \, \varphi(d_-)}{2} \sigma_s \, dW_s \\	
		& + \int_0^{t \wedge \tau_n}\Bigg[ a_s \frac{K \, \varphi(d_-)}{2 \sqrt{\omega_s}}  + \big( b_s^2 + c_s^2 \big)  \frac{K \Big(k^2_s- \frac{\omega_s^2}{4} - \omega_s \Big) \, \varphi(d_-) }{8 \sqrt{\omega_s}} \Bigg] \sigma_s^2 \, \, ds \\
		 = & \, C_0 + \int_0^{t \wedge \tau_n} b_s \frac{K \sqrt{\omega_s} \, \varphi(d_-)}{2} \sigma_s \, dZ_s + \int_0^{t \wedge \tau_n} c_s \frac{ K \sqrt{\omega_s} \, \varphi(d_-)}{2} \sigma_s \, dW_s, \qquad a.s.
	\end{split}
	\]
	for any $t \in [0,+\infty)$. We observe that 
	\[
	\begin{split}
		\int_0^{t \wedge \tau_n} \frac{ K^2  c_s^2 \omega_s \sigma^2_s \, \varphi^2(d_-)}{4}  \, ds & \le \frac{K^2}{8 \pi}  \int_0^{t \wedge \tau_n} \frac{c_s^2 \omega^2_s \sigma^2_s}{\omega_s} \, ds \le \frac{n K^2}{8 \pi}  \int_0^{t \wedge \tau_n} c_s^2 \omega^2_s \sigma^2_s \, ds < \infty, \quad a.s.
	\end{split}
	\]
	so that 
	\[
	\sqrt{\omega^{\tau_n}}  \, {c^{\tau_n}} \, {\sigma^{\tau_n}}  \, \varphi(d_-) \in L^2_{loc}\big([0,+\infty)\big).
	\]
	In the same manner, it can be proved that 
	\[
	\sqrt{\omega^{\tau_n}} \, {b^{\tau_n}} \, {\sigma^{\tau_n}} \,  \varphi(d_-) \in L^2_{loc}\big([0,+\infty)\big).
	\]
	Therefore, $C^{\tau_n}$ is a local martingale because it is equal to the sum of a constant and two local martingales. Applying twice the Put-Call-Parity (as in the proof of Lemma~\ref{lemma:locmart}), it can be shown that $C^{\tau_n}$ is also a martingale. Finally, we have that
	\[
	\langle C^{\tau_n} \rangle_t = \frac{K^2}{2} \int_0^{t \wedge \tau_n} \big(b_s^2 + c_s^2 \big) \omega_s \, \varphi^2(d_-) \sigma_s^2 \, ds, \qquad a.s.,
	\]
	that is $\langle C^{\tau_n} \rangle$ is an almost surely absolutely continuous function with respect to time because $\big((b^{\tau_n})^2 + (c^{\tau_n})^2 \big) \omega^{\tau_n} \, \varphi^2(d_-) \, (\sigma^{\tau_n})^2 \in L^1_{loc}([0,+\infty))$.
\end{proof}

$C^{\tau_n}$ in Theorem~\ref{thr:sandwichedres} is an example of what we call a sandwiched martingale:
\begin{definition}
	On $\big(\Omega,\mathcal{F},(\mathcal{F}_t)_{t \ge 0}, \mathbb{Q}\big)$, a martingale $M$ is said to be a \emph{sandwiched martingale} if there exist two progressively measurable processes $l$ and $u$ such that almost surely $l \le u$ and 
	\[
	\begin{cases}
		l_0 < M_0 < u_0 \\
		l_t \le M_t \le u_t
	\end{cases}
	\]
	for any $t > 0$.
\end{definition}

Theorem~\ref{thr:sandwichedres} stresses that it is sufficient that the call price process is described by the sandwiched martingale $C^{\tau_n}$ for the corresponding IRV to follow a locally consistent dynamics, whereas Theorem~\ref{thr:masterSDE} starts from the assumption that $\omega^{\tau_n}$ satisfies Eq.~\eqref{eq:SDE}. Significantly, Definition~\ref{def:nostaticarb} ensures that call price processes on arbitrage-free markets with smiles are described by sandwiched martingales with boundaries no broader than those in Theorem~\ref{thr:sandwichedres}, guaranteeing the local consistency of the corresponding IRV dynamics. We also note that if it is almost surely $\lim_{n \rightarrow +\infty} \tau_n = \lim_{n \rightarrow +\infty} \tau^0_n = T$, then the corresponding IRV dynamics is globally consistent. These results are consistent with those in Section~\ref{sec:masterSDE} and confirm the relationship between sandwiched martingales, locally consistent dynamics and arbitrage-free markets.

It should be clear now that identifying sandwiched martingales describing call prices is equivalent to finding solutions to the IRV master SDE. Below we provide an example of sandwiched martingale for a single option market, while in Appendix~\ref{app:exsan} we deal with a three-options market.

\begin{example}[Single option market]
	Let $\sigma$ satisfy the Novikov condition and $S$ be defined as
	\[
	\begin{cases}
		S_t = S_0 \cdot \exp\Big\{ \int_0^t \sigma_s \, dW_s - \frac{1}{2} \int_0^t \sigma_s^2 \, ds \Big\} \\
		S_0 > 0
	\end{cases}
	\]
	for any $t \in [0,T]$. It follows that $S$ is a martingale. Consider a positive martingale $N$, bounded by $1$ on $[0,T]$ and independent of $S$. For example, let
	\[
	N_t \coloneqq \exp\bigg\{ - \bigg[ T + (Z_t^2 - t ) + 2 \int_0^t Z_s^2 \, ds\bigg] \bigg\},
	\]
	where $Z$ is a Brownian motion independent of $W$. Assume $S_0 \neq K$ and define the stopping time 
	\[
	\tau \coloneqq \inf \big\{ t \in [0,T]: S_t = K\}.
	\]
	Then, the stopped process 
	\begin{equation}\label{eq:sanC}
	C^{\tau}_t \coloneqq \big(S^{\tau}_t - K \big)_+ + N^{\tau}_t \cdot  \Big[ S^{\tau}_t - \big(S^{\tau}_t - K \big)_+ \Big]
	\end{equation}
	is a sandwiched martingale such that
	\[
	\big(S^{\tau}_t - K \big)_+ < C^{\tau}_t < S^{\tau}_t
	\]
	for any $t \in [0,T]$. Indeed:
	\begin{itemize}
		\item if $S_0 > K$, then $C^{\tau}$ in Eq.~\eqref{eq:sanC} is given by
		\[
		C^{\tau}_t = S^{\tau}_t - \big(1 - N^{\tau}_t \big) K \in \big( S_t^{\tau}-K, S_t^{\tau} \big)
		\]
		for any $t \in [0,T]$, and is a martingale because it is sum of two independent martingales and a constant;
		\item if $S_0 < K$, then $C^{\tau}$ in Eq.~\eqref{eq:sanC} is given by
		\[
		C^{\tau}_t = N^{\tau}_t S^{\tau}_t \in \big( 0, S_t^{\tau} \big)
		\]
		for any $t \in [0,T]$, and is a martingale (w.r.t. its natural filtration) because product of two independent martingales (each w.r.t. its natural filtration).
	\end{itemize}
\end{example}

\begin{remark}
	In an arbitrage-free market, the price process of a call option with strike $K_1 < K_2$ must be a sandwiched martingale with values in the interval 
	\[
	\big( (S_t - K_1)_+, S_t \big) \subset \big( (S_t - K_2)_+, S_t \big)
	\]
	for any $t \in [0,T)$. It follows that if we observe a sandwiched martingale lying between $(S_t - K_1)_+$ and $S_t$ for any $t \in [0,T)$, we cannot say which is the option whose price process we are observing.
\end{remark}

%% file: Conclusion/Conclusion.tex
\section{Conclusion}
\label{sec:conclusion}

Starting from some issues in the paper by~\textcite{CarrSun:2014},  we have developed a rather general framework allowing Implied Remaining Variance models to replace standard option pricing models. According to the new framework, absence of arbitrage is related to the notion of locally (and globally) consistent dynamics. In particular, it means that the IRV process satisfies the IRV master SDE. We have mapped the findings of~\textcite{SchweizerWissel:2008b} into the new IRV framework, identifying a sufficient condition for the existence of a unique globally consistent dynamics and providing explicit family of solutions. Notwithstanding, the IRV master SDE remains an interesting object per se. In markets with a smile, we have stressed the distinction between static and dynamic arbitrage. While globally consistent dynamics ensure the absence of both types of arbitrage, no static arbitrage conditions must be imposed when the IRV dynamics is just locally consistent. These findings are consistent with the paper by~\textcite{ElAmraniJacquierMartini:2020}. The notion of smile bubble is embedded into the new framework, too. Finally, we have shown the equivalence between IRV processes following a locally consistent dynamics and option prices described by sandwiched martingales. In future research, we plan to extend the new framework to a market on which options with different maturities are traded.

%% file: Appendix/Appendix.tex
\begin{appendices}

\section{Sandwiched martingales for a three-option market}
\label{app:exsan}

In this appendix, we consider a market on which are traded three call options with strike prices $0 < K_1 < K_2 < K_3 < +\infty$ and common maturity $T > 0$. We aim to provide an example of sandwiched martingales satisfying all the no-static-arbitrage conditions in Definition~\ref{def:nostaticarb}. For notational simplicity, we denote $C_i = C(K_i)$, with $i \in \{1,2,3\}$.

Let $\sigma$ satisfy the Novikov condition and $S$ be defined as
\begin{equation}\label{eq:Sapp}
\begin{cases}
	S_t = S_0 \cdot \exp\Big\{ \int_0^t \sigma_s \, dW_s - \frac{1}{2} \int_0^t \sigma_s^2 \, ds \Big\} \\
	S_0 > 0
\end{cases}
\end{equation}
for any $t \in [0,T]$. It follows that $S$ is a martingale. 

Assume $K_3 < S_0 $ and define $\tau \coloneqq \inf \big\{ t \in [0,T]: S_t = K_3 \big\}$. By condition 2 of Definition~\ref{def:nostaticarb}, it must be
\[
-1 < \frac{C^{\tau}_{1, t} - C^{\tau}_{2,t}}{K_1-K_2} < 0 \qquad \Longleftrightarrow \qquad 0 < C^{\tau}_{1,t} - C^{\tau}_{2,t} < K_2 - K_1
\]
so that we can set
\[
C^{\tau}_{1,t} - C^{\tau}_{2,t} \equiv N^{\tau}_{12,t} \cdot (K_2 - K_1)
\]
where $N_{12}$ is a unit-bounded positive martingale, independent of $S$.

By conditions 2 and 3 of Definition~\ref{def:nostaticarb}, it must be
\[
-1 < \frac{C^{\tau}_{1,t} - C^{\tau}_{2,t}}{K_1 - K_2} \leq \frac{C^{\tau}_{2,t} - C^{\tau}_{3, t}}{K_2 - K_3} < 0 \qquad \Longleftrightarrow \qquad 0 < C^{\tau}_{2, t} - C^{\tau}_{3, t} \leq \frac{K_3 - K_2}{K_2 - K_1} \cdot \big( C^{\tau}_{1,t} - C^{\tau}_{2, t}\big) < K_3 - K_2
\]
so that we can set
\[
C^{\tau}_{2, t} - C^{\tau}_{3, t} \equiv N^{\tau}_{23, t} \cdot N^{\tau}_{12, t} \cdot (K_3 - K_2)
\]
where $N_{23}$ is a unit-bounded positive martingale, independent of $S$ and $N_{12}$.

By condition 1 of Definition~\ref{def:nostaticarb}, it must be
\[
S^{\tau}_t - K_1 < C^{\tau}_{1, t} < S^{\tau}_t
\]
so that we can set, as in the single option case,
\begin{equation}\label{eq:C1app}
C^{\tau}_{1, t} \equiv S^{\tau}_t - \big(1 - N^{\tau}_{1,t}\big) \cdot K_1
\end{equation}
where $N_{1}$ is a unit-bounded positive martingale, independent of $S$ and $N_{23}$. It follows that
\begin{align}
	\label{eq:C2app} C^{\tau}_{2,t} & = C^{\tau}_{1,t} - N^{\tau}_{12, t} \cdot (K_2 - K_1) = S^{\tau}_t - \big(1 - N^{\tau}_{1,t}\big) \cdot K_1 - N^{\tau}_{12, t} \cdot (K_2 - K_1) \\
	C^{\tau}_{3, t} & = C^{\tau}_{2,t} - N^{\tau}_{23, t} \cdot N^{\tau}_{12, t} \cdot (K_3 - K_2) \notag \\
	\label{eq:C3app} & = S^{\tau}_t - \big(1 - N^{\tau}_{1,t}\big) \cdot K_1 - N^{\tau}_{12, t} \cdot (K_2 - K_1) - N^{\tau}_{23, t} \cdot N^{\tau}_{12, t} \cdot (K_3 - K_2).	
\end{align}
Since 
\[
\begin{split}
	S^{\tau}_t - K_2 & = S^{\tau}_t - K_1 - K_2 + K_1 \\
	& = S^{\tau}_t - K_1 - (K_2 - K_1) \\
	& < S^{\tau}_t - K_1 - N^{\tau}_{12,t} \cdot (K_2 - K_1) \\
	& < C^{\tau}_{1,t}  - N^{\tau}_{12,t} \cdot (K_2 - K_1) \\
	& = C^{\tau}_{2,t} < S_t^{\tau}
\end{split}
\]
condition 1 of Definition~\ref{def:nostaticarb} holds for $C^{\tau}_{2}$. Analogously
\[
\begin{split}
	S^{\tau}_t - K_3 & = S^{\tau}_t - K_2 - K_3 + K_2 \\
	& = S^{\tau}_t - K_2 - (K_3 - K_2) \\
	& < S^{\tau}_t - K_2 - N^{\tau}_{23,t} \cdot N^{\tau}_{12,t} \cdot (K_3 - K_2) \\
	& < C^{\tau}_{2,t} - N^{\tau}_{23,t} \cdot N^{\tau}_{12,t} \cdot (K_3 - K_2)  \\
	& = C^{\tau}_{3,t} < S_t^{\tau}
\end{split}
\]
which implies that condition 1 of Definition~\ref{def:nostaticarb} holds for $C^{\tau}_{3}$, too. Condition 4 of Definition~\ref{def:nostaticarb} requires
\[
-1 < \frac{C^{\tau}_{1,t} - S^{\tau}_t}{K_1} \leq \frac{C^{\tau}_{2,t} - C^{\tau}_{1,t}}{K_2 - K_1} < 0 \qquad \Longleftrightarrow \qquad N^{\tau}_{1,t} \leq 1 - N^{\tau}_{12,t}.
\]
Finally, we note that $C^{\tau}_1, \, C^{\tau}_2$ and $C^{\tau}_3$ are martingales because they are equal to the sum and product of independent martingales (w.r.t. their natural filtration).

Therefore, assume $K_3 < S_0$ and $N_{1,0} < 1 - N_{12,0}$, and define the stopping time
\[
\xi \coloneqq \tau \wedge \big\{ t \in [0,T]: N_{1,t} = 1 - N_{12,t}\}.
\]
The processes $C^{\xi}_{1,t} \coloneqq C_{1, t \wedge \xi}$, $C^{\xi}_{2,t} \coloneqq C_{2, t \wedge \xi}$, $C^{\xi}_{3,t} \coloneqq C_{3, t \wedge \xi}$, defined as in Eq.~\eqref{eq:C1app}--\eqref{eq:C3app}, are sandwiched martingales satisfying all the no-static-arbitrage conditions of Definition~\ref{def:nostaticarb}. 

For different values of $S_0$, the same reasoning can be applied to provide examples of analogous sandwiched martingales.

\end{appendices}